\documentclass[
	journal,
]{IEEEtran}

\IEEEoverridecommandlockouts 
\makeatletter
\def\markboth#1#2{\def\leftmark{\@IEEEcompsoconly{\sffamily}\MakeUppercase{\protect#1}}%
\def\rightmark{\@IEEEcompsoconly{\sffamily}\MakeUppercase{\protect#2}}}
\makeatother

\usepackage[amssymb]{SIunits}
\usepackage[cmex10]{amsmath} 
\usepackage{amssymb, amsthm}
\usepackage{bbm} 
\usepackage[latin1]{inputenc}
\usepackage{xcolor}
\usepackage[noadjust]{cite}

\usepackage[pdftex]{graphicx}
\usepackage[caption=false, font=footnotesize]{subfig}
\usepackage{url}
\usepackage{xspace} 

\usepackage{booktabs} 

\usepackage{array}
\newcolumntype{L}[1]{>{\raggedright\let\newline\\\arraybackslash\hspace{0pt}}m{#1}}
\newcolumntype{C}[1]{>{\centering\let\newline\\\arraybackslash\hspace{0pt}}m{#1}}
\newcolumntype{R}[1]{>{\raggedleft\let\newline\\\arraybackslash\hspace{0pt}}m{#1}}

\usepackage[
	shortcuts,
	nonumberlist,	
	acronym, 
	hyperfirst=true,
	section=section, 
	order=letter,
	numberedsection=nolabel 
]{glossaries}

\theoremstyle{definition}

\theoremstyle{plain}
\newtheorem{lemma}{Lemma}

\newtheorem{theorem}{Theorem}

\theoremstyle{remark} 
\newtheorem{remark}{Remark}
\newtheorem{example}{Example}

\newcommand\xqed[1]{%
\leavevmode\unskip\penalty9999 \hbox{}\nobreak\hfill
\quad\hbox{#1}}
\newcommand\demo{\xqed{$\triangle$}}

\graphicspath{%
	{../pstricks/}%
	{../mafig/}%
}%
\newcommand{\IID}{\gls{iid}\xspace}	

\newcommand{\define}{\triangleq}
\newcommand{\D}{\mathrm{d}}

\newcommand{\VNs}{\glspl{vn}\xspace}
\newcommand{\CNs}{\glspl{cn}\xspace}

\newcommand{\E}[1]{\mathbb{E}\left[#1\right]}

\newcommand{\eps}{\varepsilon}

\newcommand{\Cbch}{\ensuremath{\mathcal{C}}}

\newcommand{\nbch}{\ensuremath{n}}
\newcommand{\kbch}{\ensuremath{k}_\mathcal{C}}
\newcommand{\dmin}{\ensuremath{d_\text{min}}}

\newcommand{\nham}{\ensuremath{n_\mathcal{C}}}
\newcommand{\kham}{\ensuremath{k_\mathcal{C}}}

\newcommand{\nhpc}{\ensuremath{m}}

\newcommand{\ErdRen}{Erd\H{o}s--R\'enyi\xspace}

\newcommand{\Gnp}{\ensuremath{\mathcal{G}(n,p)}}
\newcommand{\Gncn}{\ensuremath{\mathcal{G}(n,c/n)}}

\newcommand{\Gnk}{\ensuremath{\mathcal{G}^{\mathcal{V}}(n,
\boldsymbol{\kappa})}}

\newcommand{\bigo}{\ensuremath{\mathcal{O}}}

\newcommand{\Dec}{\ensuremath{\mathcal{D}}}

\newcommand{\Lipschitz}{\ensuremath{\Lambda}}

\newcommand{\Pois}[1]{\ensuremath{\text{\textnormal{\textsf{Po}}}(#1)}}
\newcommand{\Binom}[2]{\ensuremath{\text{\textnormal{\textsf{Bin}}}(#1,#2)}}
\newcommand{\Bern}[1]{\ensuremath{\text{\textnormal{\textsf{B}}}(#1)}}

\renewcommand{\Pr}[1]{{\mathbb{P}\left( #1 \right)}}			 

\newcommand{\tee}{\mathsf{t}}

\newcommand{\ttime}{t}

\newcommand{\nl}{}

\newcommand{\Gpc}{\ensuremath{\mathcal{C}_n(\etab, \vect{\gamma}, \vect{\tau})}}

\newcommand{\No}{\mathbb{N}_0}

\newcommand{\bRV}{\theta} 

\newcommand{\inprobto}{\ensuremath{\xrightarrow{\text{P}}}}
\newcommand{\indistrto}{\ensuremath{\xrightarrow{\text{d}}}}

\newcommand{\tmax}{\ensuremath{\tee_{\text{max}}}}
\newcommand{\tmin}{\ensuremath{\tee_{\text{min}}}}

\newcommand{\transpose}{\intercal}

\newcommand{\GwZ}{\bar{Z}}
\newcommand{\GwXi}{\bar{\xi}}
\newcommand{\GwX}{\bar{X}}
\newcommand{\GwJ}{\bar{J}}
\newcommand{\GwA}{\bar{A}}

\newcommand{\GwT}{\bar{T}}

\newcommand{\Loss}{\ensuremath{\mathcal{L}}}

\newcommand{\GPCs}{\glspl{gpc}\xspace}

\newcommand{\vect}[1]{\ensuremath{\boldsymbol{#1}}}
\newcommand{\mat}[1]{\ensuremath{\mathbf{#1}}}
\newcommand{\etab}{\ensuremath{\boldsymbol{\eta}}}

\newcommand{\cthr}{\ensuremath{{c}^*}}

\newcommand{\NoRev}[1]{%
	{
	#1%
	}%
}%

\newcommand{\RevA}[1]{%
	{
	#1%
	}%
}%

\newcommand{\RevB}[1]{%
	{
	#1%
	}%
}%

\newif\ifshow
\showfalse

\newcommand{\abbr}[1]{{#1}}				

\makeatletter
\let\aclOLD=\acl
\renewcommand{\acl}[1]{%
  \begingroup    
  \let\@@underline=\relax
  \aclOLD{#1}%
  \endgroup
}
\makeatother

\newcommand{\NewA}[3]{
	\newacronym{#1}{#2}{#3}
}

\NewA{af}{AF}{\abbr{a}mplify-and-\abbr{f}orward}
\NewA{apsk}{APSK}{\abbr{a}mplitude \abbr{p}hase-\abbr{s}hift \abbr{k}eying}
\NewA{ask}{ASK}{\abbr{a}mplitude-\abbr{s}hift \abbr{k}eying}
\NewA{ase}{ASE}{\abbr{a}mplified \abbr{s}pontaneous \abbr{e}mission}
\NewA{awgn}{AWGN}{\abbr{a}dditive \abbr{w}hite \abbr{G}aussian \abbr{n}oise}
\NewA{biawgn}{BI-AWGN}{binary-input \abbr{a}dditive \abbr{w}hite \abbr{G}aussian \abbr{n}oise}
\NewA{bep}{BEP}{\abbr{b}it \abbr{e}rror \abbr{p}robability}
\NewA{cep}{CEP}{\abbr{c}odeword \abbr{e}rror \abbr{p}robability}
\NewA{ber}{BER}{\abbr{b}it \abbr{e}rror \abbr{r}ate}
\NewA{qap}{QAP}{\abbr{q}uadratic \abbr{a}assignment \abbr{p}roblem}
\NewA{bicm}{BICM}{\abbr{b}it-\abbr{i}nterleaved \abbr{c}oded \abbr{m}odulation}				
\NewA{cm}{CM}{\abbr{c}oded \abbr{m}odulation}				
\NewA{qpsk}{QPSK}{quadrature phase-shift keying}				
\NewA{mlcm}{MLCM}{multilevel coded modulation}				
\NewA{bpsk}{BPSK}{\abbr{b}inary \abbr{p}hase-\abbr{s}hift \abbr{k}eying}				
\NewA{bsc}{BSC}{\abbr{b}inary \abbr{s}ymmetric \abbr{c}hannel}				
\NewA{brgc}{BRGC}{\abbr{b}inary \abbr{r}eflected \abbr{G}ray \abbr{c}ode}				
\NewA{cf}{CF}{\abbr{c}haracteristic \abbr{f}unction}
\NewA{csit}{CSIT}{\abbr{c}annnel \abbr{s}tate \abbr{i}nformation at the \abbr{transmitter}}
\NewA{csi}{CSI}{\abbr{c}annnel \abbr{s}tate \abbr{i}nformation}
\NewA{df}{DF}{\abbr{d}ecode-and-\abbr{f}orward}
\NewA{fd}{FD}{\abbr{f}ull-\abbr{d}uplex}
\NewA{fft}{FFT}{\abbr{f}ast \abbr{F}ourier \abbr{t}ransform}
\NewA{iid}{i.i.d.}{independent and \abbr{i}dentically \abbr{d}istributed}
\NewA{isi}{ISI}{\abbr{i}nter\abbr{s}ymbol \abbr{i}nterference}
\NewA{lb}{LB}{\abbr{l}ower \abbr{b}ound}
\NewA{map}{MAP}{\abbr{m}aximum \abbr{a} \abbr{p}osteriori}
\NewA{mf}{MF}{\abbr{m}odulo-and-\abbr{f}orward}
\NewA{mlan}{MLAN}{\abbr{m}odulo-\abbr{l}attice \abbr{a}dditive \abbr{n}oise}
\NewA{ml}{ML}{\abbr{m}aximum \abbr{l}ikelihood}
\NewA{mmse}{MMSE}{\abbr{m}inimum \abbr{m}ean \abbr{s}quare \abbr{e}rror}
\NewA{nlpc}{NLPC}{\abbr{n}onlinear \abbr{p}hase \abbr{c}ompensation}
\NewA{nlpn}{NLPN}{\abbr{n}onlinear \abbr{p}hase \abbr{n}oise}
\NewA{nc}{NC}{\abbr{n}etwork \abbr{c}oding}
\NewA{ofmd}{OFDM}{\abbr{o}rthogonal \abbr{f}requency-\abbr{d}ivision \abbr{m}ultiplexing}			
\NewA{dp}{DP}{\abbr{d}ual-\abbr{p}olarization}
\NewA{pam}{PAM}{\abbr{p}ulse \abbr{a}mplitude \abbr{m}odulation}
\NewA{pdf}{PDF}{\abbr{p}robability \abbr{d}ensity \abbr{f}unction}
\NewA{plnc}{PNC}{\abbr{p}hysical-layer \abbr{n}etwork \abbr{c}oding}
\NewA{psk}{PSK}{\abbr{p}hase-\abbr{s}hift \abbr{k}eying}
\NewA{pmd}{PMD}{\abbr{p}olarization \abbr{m}ode \abbr{d}ispersion}
\NewA{pm}{PM}{\abbr{p}olarization-\abbr{m}ultiplexed}
\NewA{pdm}{PDM}{\abbr{p}olarization-\abbr{d}ivision \abbr{m}ultiplexing}
\NewA{qam}{QAM}{\abbr{q}uadrature \abbr{a}mplitude \abbr{m}odulation}
\NewA{sqp}{SQP}{\abbr{s}equential \abbr{q}uadratic \abbr{p}rogramming}
\NewA{rd}{RD}{\abbr{r}adii \abbr{d}istribution}
\NewA{sep}{SEP}{\abbr{s}ymbol \abbr{e}rror \abbr{p}robability}
\NewA{ser}{SER}{\abbr{s}ymbol \abbr{e}rror \abbr{r}ate}
\NewA{si}{SI}{\abbr{s}ide \abbr{i}nformation}
\NewA{sp}{SP}{\abbr{s}ingle-\abbr{p}olarization}
\NewA{sanr}{SNR}{\abbr{s}ignal-to-(additive-)\abbr{n}oise \abbr{r}atio}
\NewA{snr}{SNR}{\abbr{s}ignal-to-\abbr{n}oise \abbr{r}atio}
\NewA{snlse}{sNLSE}{\abbr{s}tochastic \abbr{n}onlinear \abbr{S}chr\"odinger \abbr{e}quation}
\NewA{nlse}{NLSE}{\abbr{n}onlinear \abbr{S}chr\"odinger \abbr{e}quation}
\NewA{stwrc}{sTRC}{\abbr{s}eparated \abbr{t}wo-way \abbr{r}elay \abbr{c}hannel}
\NewA{stwtrc}{sTTRC}{\abbr{s}eparated \abbr{t}wo-way \abbr{t}wo-\abbr{r}elay \abbr{c}hannel}
\NewA{ts}{TS}{\abbr{t}wo-\abbr{s}tage}
\NewA{twrc}{TRC}{\abbr{t}wo-way \abbr{r}elay \abbr{c}hannel}
\NewA{twtrc}{TTRC}{\abbr{t}wo-way \abbr{t}wo-\abbr{r}elay \abbr{c}hannel}
\NewA{wdm}{WDM}{\abbr{w}avelength-\abbr{d}ivision \abbr{m}ultiplexing}
\NewA{vn}{VN}{variable node}
\NewA{cn}{CN}{constraint node}
\NewA{de}{DE}{density evolution}
\NewA{sc}{SC}{spatially-coupled}
\NewA{ldpc}{LDPC}{low-density parity-check}
\NewA{bp}{BP}{belief propagation}
\NewA{dm}{DM}{dispersion-managed}
\NewA{spm}{SPM}{self-phase modulation}
\NewA{xpm}{XPM}{cross-phase modulation}
\NewA{fwm}{FWM}{four-wave-mixing}
\NewA{ixpm}{IXPM}{intrachannel cross-phase modulation}
\NewA{ifwm}{IFWM}{intrachannel four-wave-mixing}
\NewA{ssfm}{SSFM}{split-step Fourier method}
\NewA{fec}{FEC}{forward error correction}
\NewA{psd}{PSD}{power spectral density}
\NewA{ook}{OOK}{on-off keying}
\NewA{smf}{SMF}{single-mode fiber}
\NewA{ssmf}{SSMF}{standard single-mode fiber}
\NewA{dbp}{DBP}{digital backpropagation}
\NewA{edfa}{EDFA}{erbium-doped fiber amplifier}
\NewA{ofdm}{OFDM}{orthogonal frequency division multiplexing}
\NewA{exit}{EXIT}{extrinsic information transfer}
\NewA{pexit}{P-EXIT}{protograph extrinsic information transfer}
\NewA{osnr}{OSNR}{optical signal-to-noise ratio}
\NewA{roadm}{ROADM}{reconfigurable optical add-drop multiplexer}
\NewA{rps}{RPS}{Raman pump station}
\NewA{mlse}{MLSE}{maximum likelihood sequence estimation}
\NewA{dcf}{DCF}{dispersion compensating fiber}
\NewA{tcm}{TCM}{trellis coded modulation}
\NewA{edc}{EDC}{electronic dispersion compensation}
\NewA{ldpcc}{LDPCC}{low-density parity-check convolutional}
\NewA{llr}{LLR}{log-likelihood ratio}
\NewA{ra}{RA}{repeat-accumulate}
\NewA{ira}{IRA}{irregular-repeat-accumulate}
\NewA{ara}{ARA}{accumulate-repeat-accumulate}
\NewA{mi}{MI}{mutual information}
\NewA{vdmm}{VDMM}{variable degree matched mapping}
\NewA{gmi}{GMI}{generalized mutual information}
\NewA{wd}{WD}{windowed decoder}
\NewA{gn}{GN}{Gaussian noise}
\NewA{hd}{HD}{hard-decision}
\NewA{hdd}{HDD}{hard-decision decoding}
\NewA{sd}{SD}{soft-decision}
\NewA{sdd}{SDD}{soft-decision decoding}
\NewA{emp}{EMP}{extrinsic message passing}
\NewA{imp}{IMP}{intrinsic message passing}
\NewA{gldpc}{GLDPC}{generalized low-density parity-check}
\NewA{scgldpc}{SC-GLDPC}{spatially-coupled generalized low-density parity-check}
\NewA{scldpc}{SC-LDPC}{spatially-coupled low-density parity-check}
\NewA{scfec}{SC-FEC}{spatially-coupled forward error correction}
\NewA{tbbc}{TBBC}{tightly-braided block code}
\NewA{gpc}{GPC}{generalized product code}
\NewA{otn}{OTN}{optical transport network}
\NewA{bch}{BCH}{Bose--Chaudhuri--Hocquenghem}
\NewA{rs}{RS}{Reed--Solomon}
\NewA{bbbc}{BBBC}{block-wise braided block code}
\NewA{hpc}{HPC}{half-product code}
\NewA{pc}{PC}{product code}
\NewA{rv}{RV}{random variable}
\NewA{mbp}{MBP}{multiple block permutator}
\NewA{cer}{CER}{codeword error probability}
\NewA{oh}{OH}{overhead}
\NewA{bdd}{BDD}{bounded-distance decoding}
\NewA{ncg}{NCG}{net coding gain}
\NewA{gwbp}{Galton--Watson branching process}{Galton--Watson branching process}
\NewA{rhs}{RHS}{right-hand side}
\NewA{lhs}{LHS}{left-hand side}

\newacronym[%
	longplural={binary erasure channels},%
	shortplural={BECs}%
]{bec}{BEC}{binary erasure channel}%

\begin{document}

\title{Density Evolution for Deterministic \\ Generalized Product Codes
on the Binary\\ Erasure Channel \RevA{at High Rates}}


\author{%
	Christian Häger,~\IEEEmembership{Student Member, IEEE},
	Henry D.~Pfister,~\IEEEmembership{Senior Member, IEEE},\\
	Alexandre Graell i Amat,~\IEEEmembership{Senior Member, IEEE}, and
	Fredrik Brännström,~\IEEEmembership{Member, IEEE}
\thanks{This work was partially funded by the Swedish Research Council
under grant \#2011-5961. \NoRev{Parts of this paper were presented at
the \emph{Optical Fiber Communication Conference, Anaheim, CA, 2016},
the \emph{International Symposium on Turbo Codes and Iterative
Information Processing, Brest, France, 2016}, and the \emph{IEEE
International Symposium on Information Theory, Barcelona, Spain,
2016}.} }%
\thanks{\NoRev{This work was conducted when C.~Häger was with the Department of Signals and Systems,
Chalmers University of Technology, SE-41296 Gothenburg, Sweden. He is now
with the Department of Electrical and Computer
Engineering, Duke University, Durham, NC 27708 USA (e-mail:
ch303@duke.edu).}}
\thanks{\NoRev{H.~Pfister is with the Department of Electrical and Computer
Engineering, Duke University, Durham, NC 27708 USA (e-mail:
henry.pfister@duke.edu).}}%
\thanks{\NoRev{A.~Graell i Amat and F.~Br\"{a}nnstr\"{o}m are with
the Department of Signals and Systems, Chalmers University of Technology,
SE-41296 Gothenburg, Sweden (e-mail: \{alexandre.graell,
fredrik.brannstrom\}@chalmers.se).}}%
}%

\maketitle

\begin{abstract}
	\Glspl{gpc} are extensions of \glspl{pc} where code symbols are
	protected by two component codes but not necessarily arranged in a
	rectangular array. We consider a deterministic construction of
	\glspl{gpc} (as opposed to randomized code ensembles) and analyze
	the asymptotic performance over the \acl{bec} under iterative
	decoding. Our code construction encompasses several classes of
	\glspl{gpc} previously proposed in the literature, such as irregular
	\glspl{pc}, block-wise braided codes, and staircase codes. It is
	assumed that the component codes can correct a fixed number of
	erasures and that the length of each component code tends to
	infinity. We show that this setup is equivalent to studying the
	behavior of a peeling algorithm applied to a sparse inhomogeneous
	random graph.  Using a convergence result for these graphs, we
	derive the \acl{de} equations that characterize the asymptotic
	decoding performance. As an application, we discuss the design of
	irregular \glspl{gpc} employing a mixture of component codes with
	different erasure-correcting capabilities. 
\end{abstract}
\begin{IEEEkeywords} 
	Binary erasure channel, braided codes, density evolution, generalized low-density
	parity-check codes, inhomogeneous random graphs, multi-type
	branching processes, product codes, staircase codes. 
\end{IEEEkeywords}

\glsresetall

\section{Introduction}

Many code constructions are based on the idea of building longer codes
from shorter ones \cite{Gallager1962, ForneyJr.1965, Berrou1993a}. In
particular, \glspl{pc}, originally introduced by Elias in 1954
\cite{Elias1954}, are constructed from two linear component codes,
$\mathcal{C}_1$ and $\mathcal{C}_2$, with respective lengths $n_1$ and
$n_2$. The codewords in a \gls{pc} are rectangular $n_1 \times n_2$
arrays such that every row is a codeword in $\mathcal{C}_1$ and every
column is a codeword in $\mathcal{C}_2$. In 1981, Tanner significantly
extended this construction and introduced \gls{gldpc} codes
\cite{Tanner1981}. \gls{gldpc} codes are defined via bipartite graphs
where \VNs and \CNs represent code symbols and component code
constraints, respectively. If the underlying graph of a \gls{gldpc}
code consists exclusively of degree-2 \VNs (i.e., each code symbol is
protected by two component codes), the code is referred to as a
generalized \gls{pc} (GPC)\glsunset{gpc}. Most of the examples
presented in \cite{Tanner1981} fall into this category. 

\glspl{pc} have an intuitive iterative decoding algorithm and are used
in a variety of applications \cite{Abramson1968, Ryan2009}. In
practice, the component codes are typically \gls{bch} or \acl{rs}
codes, which can be efficiently decoded via algebraic \gls{bdd}. This
makes \GPCs particularly suited for high-speed applications due to
their significantly reduced decoding complexity compared to
message-passing decoding of \gls{ldpc} codes \cite{Smith2012a}. For
example, \glspl{gpc} have been investigated by many authors as
practical solutions for forward-error correction in fiber-optical
communication systems \cite{Justesen2007, Justesen2010, Justesen2011,
Scholten2010, Smith2012a, Jian2014, Zhang2014, Haeger2015ofc}. 

The iterative decoding of \glspl{gpc} is a standard element in many of
these systems and the analysis of iterative decoding is typically
based on \gls{de} \cite{Luby2001, Richardson2001} using an ensemble
argument.  That is, rather than analyzing a particular code directly,
one considers a set of codes, defined via suitable randomized
connections between \VNs and \CNs in the Tanner graph. Some notable
exceptions include Gallager's original analysis based on deterministic
constructions of large-girth \gls{ldpc} codes~\cite{Gallager1963},
Tanner's analysis of Hamming GPCs~\cite{Tanner1981}, the analysis of
PCs using monotone graph properties~\cite{Schwartz2005}, and the
analysis of PCs based on the $k$-core problem~\cite{Justesen2007,
Justesen2011}.

In this paper, we focus on the asymptotic performance of \glspl{gpc}
over the \gls{bec} assuming iterative decoding based on \gls{bdd} of
the component codes. In particular, we consider the case where the
component codes have a fixed erasure-correcting capability and the
length of each component code tends to infinity.
Like~\cite{Schwartz2005, Justesen2007, Justesen2011}, we consider a
\emph{deterministic} construction of \glspl{gpc}. Indeed, many classes
of \glspl{gpc} have a very regular structure in terms of their Tanner
graph and are not at all random-like. 
The code construction we
consider is sufficiently general to recover several of these classes
as special cases, such as irregular \glspl{pc}
\cite{Hirasawa1984,Alipour2012}, block-wise braided codes
\cite[Sec.~III]{Feltstrom2009}, and staircase codes \cite{Smith2012a}.
The main contribution of this paper is to show that, analogous to
\gls{de} for code ensembles, the asymptotic performance of the
considered \gls{gpc} construction is rigorously characterized by a
recursive update equation. 

Like~\cite{Schwartz2005, Justesen2007, Justesen2011}, this paper is
largely based on results that have been derived in random graph
theory. In our case, the Tanner graph itself is deterministic and
consists of a fixed arrangement of (degree-2) VNs and CNs.  Randomness
is introduced entirely due to the channel by forming the so-called
residual graph (or error graph) from the Tanner graph, i.e., after
removing known VNs and collapsing erased VNs into edges
\cite{Schwartz2005, Justesen2007, Justesen2011}. Thus, different
channel realizations give rise to an ensemble of residual graphs,
facilitating the analysis. The code construction considered here is
such that the residual graph ensemble corresponds to the sparse
inhomogeneous random graph model in \cite{Bollobas2007}. Analyzing the
decoding failure of the iterative decoder (for a fixed number of
iterations) can then be translated into a graph-theoretic question
about the behavior of a peeling algorithm applied to such a random
graph. We can then use a convergence result in \cite{Bollobas2007} to
conclude that, as the number of vertices in the graph tends to
infinity, the correct limiting behavior is obtained by evaluating the
peeling algorithm on a multi-type branching process. 

A similar connection between large random graphs and branching
processes also arises in the \gls{de} analysis for code ensembles,
e.g., irregular \gls{ldpc} codes. The main difference between this and
our setup is that, for code ensembles, the Tanner graph itself is
random due to the randomized edge connections in the ensemble
definition. \gls{de} relies on the fact that the asymptotic behavior
of an extrinsic iterative message-passing decoder can be analyzed by
considering an ensemble of computation trees
\cite[Sec.~3.7.2]{Richardson2008} (see also \cite[Sec.~1]{Luby1998b}).
This tree ensemble can alternatively be viewed as a multi-type
branching process, where types correspond to VNs and CNs of different
degrees. A tree-convergence and concentration result ensures that the
performance of a code taken (uniformly at random) from the 
ensemble will be close to the predicted \gls{de} behavior, provided
that the code is sufficiently long \cite[Th.~2]{Richardson2001}. 

The above ensemble approach can be applied to \gls{gldpc} codes and
thus also to \glspl{gpc}. For example, in
\cite{Lentmaier2010,Lentmaier2009,Lentmaier} a \gls{de} analysis for
protograph-based braided codes is presented, where the Tanner graph of
a tightly-braided code is interpreted as a protograph
\cite{Thorpe2005}. An ensemble approach has been further applied to
regular \glspl{gpc} in \cite{Miladinovic2008}, where the authors
analyze the asymptotic ensemble performance and derive the
corresponding iterative decoding thresholds. In \cite{Jian2012,
Jian2015, Zhang2015}, the authors perform a \gls{de} analysis for
\gls{gpc} ensembles paying special attention to so-called
spatially-coupled codes. \RevA{On the other hand, many \GPCs proposed
for practical systems (e.g., the recent code proposals for optical
transport networks in \cite{Smith2012a} and \cite{Jian2014}) are
entirely deterministic and not based on a randomized code ensemble.
One reason for this is that deterministic \GPCs have been shown to
achieve extremely low error floors in practice. Moreover, the inherent
code structure often results in implementation advantages compared to
randomized \GPCs. For example, the array representation of many
deterministic GPCs facilitates the use of simple hardware layouts and
efficient ``row-column'' iterative decoding schedules, whereas
ensemble-based GPCs are unlikely to possess an array representation.
Therefore, given the structured Tanner graphs of many practical
\gls{gpc} classes, it would be highly desirable to make precise
statements about the performance of actual codes, without resorting to
an ensemble argument.}

The work here is closely related to \cite{Schwartz2005, Justesen2007,
Justesen2011}. In \cite{Schwartz2005}, combinatorial tools from the
study of random graphs are used to analyze the iterative decoding of
\glspl{pc}. In \cite{Justesen2007}, the authors point out the direct
connection between the iterative decoding of \glspl{pc} and a
well-studied problem in random graph theory: the emergence of a
$k$--core, defined as the largest induced subgraph where all vertices
have degree at least $k$ \cite{Pittel1996}. Indeed, assuming that all
component codes can correct $\tee$ erasures and allowing for an
unrestricted number of iterations, the decoding either finishes
successfully, or gets stuck and the resulting graph corresponds to the
$(\tee + 1)$--core of the residual graph. The results in
\cite{Pittel1996} apply to \glspl{pc} only after some modifications
(described in \cite{Justesen2007}), since the random graph model in
\cite{Pittel1996} is slightly different than the actual one
corresponding to the residual graph ensemble of \glspl{pc}. In a later
paper, Justesen considered \glspl{gpc} for which the Tanner graph is
based on a complete graph \cite{Justesen2011} (see, e.g.,
Fig.~\ref{fig:hpc}(b)).  In that case, the results in
\cite{Pittel1996} are directly applicable. The resulting codes are
referred to as \glspl{hpc}.  Even though these codes have received
very little attention in the literature, Tanner already used a similar
construction \cite[Fig.~6]{Tanner1981}.

We use \glspl{hpc} as the starting point for our analysis. The reason
is that the residual graph of an \gls{hpc} corresponds exactly to an
instance of the \ErdRen random graph model $\Gnp$ \cite{Erdos1959,
Gilbert1959}, which is arguably one of the most well-studied random
graph models and also considerably simpler than the inhomogeneous
random graph model in \cite{Bollobas2007}. It is therefore instructive
to consider this case in sufficient detail before analyzing
generalizations to other \glspl{gpc}. Even though other classes of
\glspl{gpc} are mentioned and discussed also in \cite{Justesen2011}
(e.g., braided codes), so far, rigorous analytical results about the
asymptotic performance of deterministic \glspl{gpc} have been limited
to conventional \glspl{pc} and \glspl{hpc}. 

As an application of the derived DE equations for deterministic GPCs,
we discuss the optimization of component code mixtures for HPCs. In
particular, we consider the case where the component codes can have
different erasure-correcting capabilities. It is shown that, similar
to irregular PCs \cite{Hirasawa1984, Alipour2012}, HPCs greatly
benefit from employing component codes with different strengths,
\RevA{both in terms of decoding thresholds and finite-length performance}.  We
further derive upper and lower bounds on the iterative decoding
thresholds of HPCs with component code mixtures. The upper bound is
shown to have a graphical interpretation in terms of areas related to
the DE equations, similar to the area theorem of irregular LDPC codes. 

The remainder of the paper is structured as follows. We start by
analyzing \glspl{hpc} in Sections~\ref{sec:hpc},
\ref{sec:random_graphs}, and \ref{sec:hpc_proof}. In particular, in
Section~\ref{sec:hpc} we discuss the code construction, the decoding
algorithm, and state the main result about the asymptotic performance
of HPCs in Theorem \ref{th:hpc_result}. In
Section~\ref{sec:random_graphs}, we review the necessary background
about random graphs and branching processes related to the proof of
Theorem \ref{th:hpc_result}, which is then given in
Section~\ref{sec:hpc_proof}. In Section~\ref{sec:gpc}, we extend
Theorem \ref{th:hpc_result} to a general deterministic construction of
\glspl{gpc} and derive the corresponding DE equations. The
optimization of component code mixtures for irregular \glspl{hpc} is
studied in Section~\ref{sec:irregular_hpc}.  The paper is concluded in
Section~\ref{sec:conclusion}.

\subsection{Notation}
\label{sec:notation}

The following notation is used throughout the paper. We define the
sets $[n] \define \{1, 2, \dots, n\}$, $\No \define \{0, 1, 2,
\dots\}$, and $\mathbb{N} \define \{1, 2, \dots\}$. The cardinality of
a set $\mathcal{A}$ is denoted by $|\mathcal{A}|$. Sequences are
denoted by $(x_n)_{n \geq 1} = x_1, x_2, \ldots$. The \gls{pdf} of a
\gls{rv} $X$ is denoted by $f_X(\cdot)$.  Expectation and probability
are denoted by $\E{\cdot}$ and $\Pr{\cdot}$, respectively. We write $X
\sim \Bern{p}$ if $X$ is a Bernoulli \gls{rv} with success probability
$p$, $X \sim \Binom{n}{p}$ if $X$ is a Binomial \gls{rv} with
parameters $n$ and $p$, and $X \sim \Pois{\lambda}$ if $X$ is a
Poisson \gls{rv} with mean $\lambda$.  With some abuse of notation, we
write, e.g., $\Pr{\Pois{\lambda} \geq \tee}$ for $\Pr{X \geq \tee}$
with $X \sim \Pois{\lambda}$. We define the Poisson tail probability
as $\Psi_{\geq \tee}(\lambda) \define \Pr{\Pois{\lambda} \geq \tee} =
1 - \sum_{i=0}^{\tee-1} \Psi_{= i}(\lambda)$, where $\Psi_{=
i}(\lambda) \define \frac{\lambda^i}{i!} e^{-\lambda}$.  We use
boldface to denote vectors and matrices (e.g., $\vect{a}$ and
$\mat{A}$). Matrix transpose is denoted by $(\cdot)^\transpose$.
Convergence in distribution (weak convergence) is denoted by
$\indistrto$ and convergence in probability by $\inprobto$.  For
positive real functions, standard asymptotic notation (as $n \to
\infty$) will be used, e.g., we write $f(n) = \mathcal{O}(g(n))$ if
there exist constants $k, n_0$ such that $f(n) \leq k g(n)$ for all $n
> n_0$.  We write $f(n) = \Omega(g(n))$ if there exist constants $k,
n_0$ such that $f(n) \geq k g(n)$ for all $n > n_0$.  We write $f(n) =
\Theta( g (n))$ if both $f(n) =\mathcal{O}(g(n))$ and $f(n) =
\Omega(g(n))$.  Finally, a code is called an $(n, k, d)$ code if it is
linear and it has length $n$, dimension $k$, and minimum distance $d$.

\section{Half-Product Codes}
\label{sec:hpc}

\subsection{Code Construction}
\label{sec:hpc_construction}

Let $\Cbch$ be a binary $(n,\kbch, \tee+1)$ code and recall that such
a code can correct all erasure patterns up to weight $\tee$. An
\gls{hpc} is constructed as follows
(cf.~\cite[Sec.~III-B]{Justesen2011}). Start with a conventional
\gls{pc} defined as the set of $n \times n$ arrays such that each row
and column is a codeword in the component code $\mathcal{C}$.  Then,
form a subcode of this \gls{pc} by retaining only symmetric codeword
arrays (i.e., arrays that are equal to their transpose) with a zero
diagonal. After puncturing the diagonal and the upper (or lower)
triangular part of the array, one obtains an \gls{hpc} of length
$\nhpc = \binom{n}{2}$.  The Tanner graph representing an \gls{hpc} is
obtained from a complete graph with $n$ vertices by interpreting each
vertex as a CN corresponding to $\Cbch$ (shortened by one bit) and
replacing each of the $\nhpc$ edges by two half-edges joint together
by a \gls{vn} \cite[Sec.~III-B]{Justesen2011}.\footnote{One way to see
this is to incorporate the symmetry constraint into the Tanner graph
of a \gls{pc} by connecting each VN to the ``transposed'' VN through a
single parity-check (forcing the two to be equal). The graph now
consists of degree-3 VNs (one row, one column, and one symmetry
constraint), but can be simplified by removing all row (or column)
constraints.} In the following, we assume some fixed (and arbitrary)
ordering on the CNs and VNs. 

\begin{figure}[t]
	\centering
	\subfloat[$5 \times 5$ array]{\includegraphics{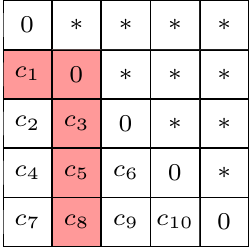}}
	\quad
	\subfloat[Tanner graph]{\includegraphics{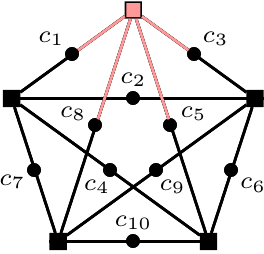}}
	\quad
	\subfloat[residual graph]{\includegraphics{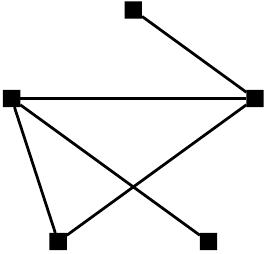}}

	\caption{Illustrations for an \gls{hpc} with $n = 5$. In the array,
	``*'' means ``equal to the transposed element''. The highlighted
	array elements illustrate one particular code constraint, which is
	also highlighted in the Tanner graph.  }
	\label{fig:hpc}
\end{figure}

\begin{example}
	\label{ex:hpc}
	Figs.~\ref{fig:hpc}(a) and (b) show the code array and Tanner graph
	of an \gls{hpc} for $n=5$ and $m = 10$. The highlighted array
	elements show the code symbols participating in the second row
	constraint, which, due to the enforced symmetry, is also the second
	column constraint.  Effectively, each component code acts on an
	L-shape in the array, i.e., both a partial row and column, which
	includes one diagonal element.  The degree of each CN is $n-1 = 4$,
	due to the zeros on the diagonal. For example, for the highlighted
	CN in Fig.~\ref{fig:hpc}(b), the second bit position of
	$\mathcal{C}$ is shortened (i.e., set to zero). Different bit
	positions are shortened for different CNs. Thus, the effective
	$(n-1, \kbch -1, \tee+1)$ component codes associated with the CNs
	are not necessarily the same.  \demo
\end{example}

\begin{remark}
	\label{rmk:bit_assignments}
	Recall that for a Tanner graph with generalized \CNs, the edges
	emanating from each CN should also be labeled with the corresponding
	component code bit positions \cite[Sec.~II]{Tanner1981}. For
	\glspl{hpc}, this assignment is implicitly given due to the array
	description. For example, the edges emanating from the highlighted
	CN in Fig.~\ref{fig:hpc}(b) correspond to bit positions $1$, $5$,
	$4$, and $3$ (in left-to-right order). Reshuffling these assignments
	may result in an overall code with different properties (e.g., rate)
	even though the Tanner graph remains unchanged
	\cite[Sec.~II]{Tanner1981}, \cite[Sec.~III-A]{Justesen2011}.
	However, for the considered iterative decoder, the performance
	remains identical as long as the component code associated with each
	CN is able to correct $\tee$ erasures, regardless of the bit
	position assignment. 
\end{remark}

We consider the limit $n \to \infty$, i.e., we use the number of \CNs
in the Tanner graph to denote the problem size as opposed to the code
length $\nhpc = \bigo(n^2)$. Assuming that $\mathcal{C}$ has a fixed
erasure-correcting capability\footnote{More precisely, we consider
sequences of codes with increasing length and fixed erasure-correcting
capability.}, this limit is sometimes referred to as the high-rate
scaling limit or high-rate regime \cite{Jian2012}. Indeed, if
$\mathcal{C}$ has dimension $\kbch$, the rate of an \gls{hpc} is
lower-bounded by \cite[Sec.~5.2.1]{Ryan2009} (see also
\cite[Th.~1]{Tanner1981}) 
\begin{align} 
	\label{eq:hpc_rate_lower_bound}
	R \geq 1 - \frac{n (n - 1 - (\kbch-1))}{m} = 1 - 2 \frac{n-\kbch}{n-1}. 
\end{align} 
For a fixed erasure-correcting capability, we can assume that
$n-\kbch$ in \eqref{eq:hpc_rate_lower_bound} stays constant. It
follows that $R \to 1$ as $n \to \infty$. Note that the dimension of
an \gls{hpc} is $\kbch(\kbch-1)/2$ \cite[Sec.~III-B]{Justesen2011},
\cite[Lem.~8]{Pfister2015}, which leads to a slightly larger rate than
the lower bound in \eqref{eq:hpc_rate_lower_bound}. 

\subsection{Binary Erasure Channel} 
\label{sec:hpc_bec}

Suppose that a codeword of an \gls{hpc} is transmitted over the
\gls{bec} with erasure probability $p$. Let $I_k$ be the number of
initial erasures associated with the $k$-th component code constraint.
Due to symmetry, we have $\E{I_k} = p (n-1)$ for all $k \in [n]$.
Moreover, using a Chernoff bound, it can be shown that $I_k$
concentrates around its mean (see, e.g.,
\cite[Sec.~IV]{Schwartz2005}). As a consequence, for a fixed $p > 0$
and $n \to \infty$, we see that any decoding attempt will be futile
since $\E{I_k} \to \infty$ for all $k$, but, on the other hand, we
assumed a finite erasure-correcting capability for the component
codes. We therefore let the erasure probability decay slowly as $p =
c/n$, for a fixed $c>0$.  Since now $p \to 0$ as $n \to \infty$, one
may (falsely) conclude that decoding will always be successful in the
asymptotic limit. As we will see, however, the answer depends
crucially on the choice of $c$.  It is thus instructive to interpret
$c$ as the ``effective'' channel quality for the chosen scaling of the
erasure probability. From the above discussion, its operational
meaning is given in terms of the expected number of initial erasures
per component code constraint for large $n$, i.e., $\E{I_k} = c
(n-1)/n \approx c$. 

\begin{remark}
One may alternatively assume a fixed erasure probability $p$, in
conjunction with sequences of component codes that can correct a fixed
fraction of erasures in terms of their block length. However, in that
case, a simple analysis reveals that the (half-)product construction
is essentially useless in the limit $n \to \infty$, and it is indeed
better to just use the component code by itself (see the discussion in
\cite[Sec.~IV]{Schwartz2005}).
\end{remark}

\subsection{Iterative Decoding} 
\label{sec:hpc_decoding}

Suppose decoding is performed iteratively for $\ell$ iterations
according to the following procedure. In each iteration, perform
\gls{bdd} for all \CNs based on the values of the connected \VNs.
Afterwards, update previously erased VNs according to the decoding
outcome. Updates are performed whenever there exists at least one CN
where the weight of the associated erasure pattern is less than or
equal to $\tee$. If the weight exceeds $\tee$, we say that the
corresponding component code declares a decoding failure. 

\begin{remark}
	\label{rmk:message_passing}
	The decoding can alternatively be interpreted as an (intrinsic)
	message-passing decoder. In the first iteration, all VNs forward the
	received channel observations to the connected CNs. Then, CNs perform
	\gls{bdd} based on all incoming messages and update their outgoing
	messages according to the decoding outcome. In subsequent
	iterations, outgoing VN messages are changed from erased to known if
	any of the two incoming CN messages becomes known. These update rules
	for VN and CN messages are not extrinsic
	(cf.~\cite[p.~117]{Richardson2008}), since the outgoing message
	along an edge may depend on the incoming message along the same
	edge. 
\end{remark}

An efficient way to represent the decoding is to consider the
following peeling procedure. First, form the residual graph from the
Tanner graph by deleting \glspl{vn} and adjacent edges associated with
correctly received bits and collapsing erased \glspl{vn} into edges
\cite{Schwartz2005, Justesen2007, Justesen2011}. Then, in each
iteration, determine all vertices that have degree at most $\tee$ and
remove them, together with all adjacent edges. The decoding is
successful if the resulting graph is empty after (at most) $\ell$
iterations. 

\begin{example}
	Fig.~\ref{fig:hpc}(c) shows the residual graph for the \gls{hpc} in
	Example \ref{ex:hpc}, where $c_2$, $c_3$, $c_4$, $c_7$, and $c_9$
	are assumed to be erased. One may check that for $\tee = 1$, the
	decoding gets stuck after one iteration while for $\tee = 2$, the
	decoding finishes successfully after two iterations. \demo
\end{example}

\begin{remark}
	The above parallel peeling procedure should not be confused with the
	sequential ``peeling decoder'' described in, e.g.,
	\cite[p.~117]{Richardson2008}. That decoder uses a different
	scheduling where vertices are removed sequentially and not
	in parallel, i.e., in each step one picks only one vertex with
	degree at most $\tee$ (uniformly at random) and removes it
	\cite[p.~117]{Richardson2008}. 
\end{remark}

\subsection{Asymptotic Performance}
\label{sec:hpc_performance}

For a fixed $\ell$, we wish to characterize the asymptotic decoding
performance as $n \to \infty$. We start by giving a heuristic argument
behind the result stated in Theorem \ref{th:hpc_result} below. For a
similar discussion in the context of cores in random graphs, see
\cite[Sec.~2]{Pittel1996}.

\begin{itemize}
	\item Consider a randomly chosen CN. The decoding outcome of the
		\gls{bdd} for this CN after $\ell$ iterations depends only on the
		depth-$\ell$ neighborhood\footnote{The depth-$\ell$ neighborhood
		of a vertex is the subgraph induced by all vertices that can be
		reached by taking $\ell$ or fewer steps from the vertex.} of the
		vertex in the residual graph corresponding to this CN. The
		residual graph itself is an instance of the \ErdRen random graph
		model $\Gnp$, which consists of $n$ vertices. An edge between two
		vertices exists with probability $p = c/n$, independently of all
		other edges. 
		
	\item For large $n$, the fixed-depth neighborhood approximately
		looks like a Poisson branching process, which starts with an
		initial vertex at depth $0$ that has a Poisson number of
		neighboring vertices with mean $c$ that extend to depth $1$. Each
		of these vertices has again a Poisson number of neighboring
		vertices, independently of all other vertices, and so on. 

	\item For large $n$ and fixed $\ell$, one would therefore expect the
		probability that an individual CN declares a failure to be close
		to the probability that the root vertex of the first $\ell$
		generations of the branching process survives the same peeling
		procedure as described for the residual graph. We define the
		latter probability as $z^{(\ell)}$. We will see in
		Section~\ref{sec:hpc_de} that
		\begin{align}
			\label{eq:hpc_de_z}
			z^{(\ell)} = \Psi_{\geq \tee+1}(c x^{(\ell-1)}),
		\end{align} 
		where the function $\Psi_{\geq \tee}$ is defined in
		Section~\ref{sec:notation} and $x^{(\ell)}$ is defined recursively
		by $x^{(0)} = 1$ and
		\begin{align}\label{eq:hpc_de_x} x^{(\ell)} = \Psi_{\geq \tee}(c
		x^{(\ell-1)}). \end{align}

\end{itemize}
The main result for \glspl{hpc} is as follows. 

\begin{theorem}
	\label{th:hpc_result}
	Let $W_k\nl$ be the indicator \gls{rv} for the event that the $k$-th
	component code declares a decoding failure after $\ell$ iterations
	of decoding and let the fraction of failed component codes be $W\nl
	= \frac{1}{n} \sum_{k=1}^n W_k\nl$. Then, we have 
\begin{align}
	\label{eq:hpc_cycle_free_bahavior}
	\lim_{n \to \infty} \mathbb{E}[W\nl] =  z^{(\ell)}.
\end{align}
Furthermore, for any $\eps \geq 0$, there exist $\delta > 0$,
$\beta>0$, and $n_0 \in \mathbb{N}$ such that for all $n > n_0$ we
have
\begin{align}
	\label{eq:hpc_concentration}
	\Pr{|W\nl - \mathbb{E}[W\nl]| \geq \eps} \leq e^{-\beta
	n^\delta}.
\end{align}
\end{theorem}

\begin{proof}
	The proof is given in Section~\ref{sec:hpc_proof}.
\end{proof}

\begin{remark}
	In our notation, we largely suppress the dependence of the involved
	RVs on $n$ and $\ell$ (e.g., one could write $W^{(n,
	\ell)}$ instead of $W$). 
\end{remark}


Combining \eqref{eq:hpc_cycle_free_bahavior} and
\eqref{eq:hpc_concentration} allows us to conclude that the code
performance after $\ell$ iterations (measured in terms of the RV $W$,
i.e., the fraction of component codes that declare failure) converges
almost surely to a deterministic value, i.e., it sharply concentrates
around $z^{(\ell)}$ for sufficiently large $n$. This result is
analogous to the \gls{de} analysis of \gls{ldpc} codes
\cite[Th.~2]{Richardson2001}, and hence, we refer to
\eqref{eq:hpc_de_z} and \eqref{eq:hpc_de_x} as the \gls{de} equations. 

The chosen performance measure in Theorem \ref{th:hpc_result} is the
most natural one for the proof in Section~\ref{sec:hpc_proof}. It is,
however, possible to relate \eqref{eq:hpc_de_z} and
\eqref{eq:hpc_de_x} to other performance measures that are more
relevant in practice. 

\begin{figure}[t]
	\begin{center}
		\includegraphics{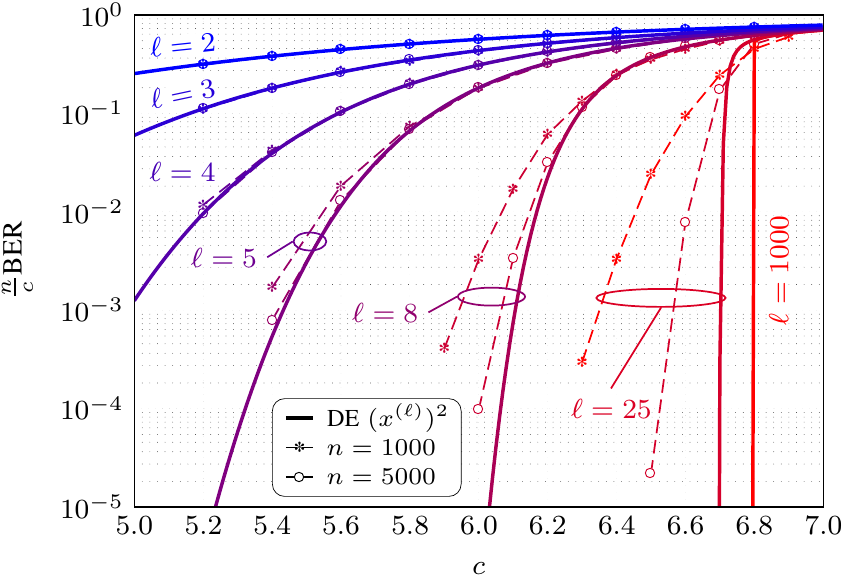}
	\end{center}
	\caption{\gls{de} and simulation results for \glspl{hpc} with $\tee
	= 4$ as a function of the iteration number $\ell$. }
	\label{fig:hpc_ber_iteration}
\end{figure}

\begin{example}
	The meaning of the quantity $x^{(\ell)}$ is given in Section
	\ref{sec:hpc_de} in terms of the Poisson branching process. The
	operational meaning in the coding context is as follows. Consider a
	randomly chosen erased bit. Asymptotically, $x^{(\ell)}$ corresponds
	to the probablity that the bit is not recovered after $\ell$
	decoding iterations by one of the two corresponding component codes.
	Since each bit is protected by two component codes, the overall
	probability of not recovering the bit is asymptotically given by
	$(x^{(\ell)})^2$. In Fig.~\ref{fig:hpc_ber_iteration}, we plot the
	resulting \gls{de} prediction $(x^{(\ell)})^2$ as a function of $c$
	for $\tee = 4$ and different values of $\ell$, together with
	simulation results of the (scaled) \gls{ber} for $n = 1000$ and $n =
	5000$. Asymptotically as $n \to \infty$, we expect the simulation
	results to converge to the solid lines. \RevA{
	The rate of
	convergence is not analyzed in this paper, e.g., through a
	finite-length scaling analysis. It should be noted, however, that
	the convergence rate in terms of the code length $m$ is rather
	``slow''. More precisely, consider the gap $\Delta$ between the DE
	prediction and finite-length simulations for $\frac{n}{c}\text{BER}
	= 10^{-3}$ and $\ell=25$ in Fig.~\ref{fig:hpc_ber_iteration}. From
	the simulation results, one may estimate that $\Delta \approx
	\bigo(n^{-1/2}) = \bigo(m^{-1/4})$, since $m = \bigo(n^2)$.} \demo
\end{example}


Theorem \ref{th:hpc_result} can be seen as an application of
\cite[Th.~11.6]{Bollobas2007}, except for the concentration bound in
\eqref{eq:hpc_concentration}. In fact, \cite[Th.~11.6]{Bollobas2007}
applies to a more general class of inhomogeneous random graphs, and we
use it later when studying generalizations to other \glspl{gpc}. The
reason for including a separate proof for \glspl{hpc} in
Section~\ref{sec:hpc_proof} is two-fold. First, since
\cite[Th.~11.6]{Bollobas2007} applies to a more general class of
random graphs, it is instructive to consider the simplest case, i.e.,
the random graph $\Gnp$ corresponding to \glspl{hpc}, separately and
in more detail. Second, rather than relying on
\cite[Th.~11.6]{Bollobas2007}, a self-contained proof of Theorem
\ref{th:hpc_result} allows us to point out similarities and
differences to the \gls{de} analysis for \gls{ldpc} codes in
\cite{Luby1998b, Richardson2001}, which we believe many readers are
familiar with. 

As mentioned in~\cite{Justesen2011}, iterative decoding of \glspl{hpc}
over the \gls{bec} is closely related to the emergence of a $k$--core
in $\Gnp$. First observe that the overall decoding is successful if
the \gls{rv} $W$ is strictly zero, i.e., if none of the component
decoders declare failure. The existence of a core can then be related
to the overall decoding failure assuming an unrestricted number of
iterations. Therefore, there is a subtle difference between studying
the core and the overall decoding failure in our setup. In our case,
the notion of decoding failure is always linked to the number of
decoding iterations, which is assumed to be fixed
(cf.~\cite[Sec.~3.19]{Richardson2008}). As a consequence, even though
the overall decoding may fail after a finite number of iterations,
there need not be a core in the residual graph. (The decoding may have
been successful if we had done one more iteration, say.) Linking the
decoding failure to the number of iterations has the advantage that it
can always be determined locally (within the neighborhood of each
vertex), whereas the core is a global graph property.  In general,
additional effort is required to infer information about global graph
properties from local ones \cite[Sec.~3.3]{Bollobas2009a},
\cite{Riordan2007}.   

\RevA{%
\subsection{Performance Prediction of Finite-Length Codes}

Before proving Theorem \ref{th:hpc_result}, it is instructive to
discuss the practical implications of the asymptotic DE analysis for
finite-length codes. This is particularly important because the
high-rate scaling limit implies $p \to 0$ and $R \to 1$ as $n \to
\infty$, which seems to preclude any practical usefulness. To see that
this is not the case, we start by reviewing the practical usefulness
of DE for finite-length LDPC codes. 

DE is typically used to find decoding thresholds that divide the
channel quality parameters range (e.g., the erasure probability or the
signal-to-noise ratio) into a region where reliable communication is
possible and where it is not. This interpretation of the threshold as
a sharp dividing line is appropriate for $n \to \infty$, where, for
LDPC codes, $n$ is the code length. On the other hand, for finite $n$,
thresholds are still useful to approximately predict the region of the
channel quality parameter range where the performance curve of a
finite-length LDPC code bends into the characteristic waterfall
behavior. Moreover, thresholds have been used with great success as an
optimization criterion to improve the performance of practical,
finite-length LDPC codes in a wide variety of applications.  The
rationale behind this approach is that threshold improvements
translate quite well into performance improvements, at least if $n$ is
sufficiently large.  While there is no guarantee that this approach
works, it typically leads to fast and efficient optimization routines. 

The asymptotic DE analysis in this paper can be used in essentially
the same way. The main conceptual difference with respect to DE for
LDPC codes is that decoding thresholds are not given in terms of the
actual channel quality, but rather in terms of the effective channel
quality. In particular, the decoding threshold is formally defined as
\begin{align}
	\label{eq:gpc_threshold}
	\cthr \define \sup\{ c > 0 \,|\, \lim_{\ell \to \infty} z^{(\ell)} = 0
	\}.
\end{align}
Asymptotic results, including thresholds, can be translated into a
nonasymptotic setting by considering the erasure probability scaling
$p = c/n$ for a given (finite) $n$. For example, consider the
results for HPCs with $\tee = 4$ shown in
Fig.~\ref{fig:hpc_ber_iteration}. The threshold in this case is
located at approximately $\cthr \approx 6.8$. Therefore, we should
expect the waterfall behavior for $n = 1000$ to start at $p \approx
0.0068$ and for $n = 5000$ at $p \approx 0.00136$. 

The practical usefulness of the asymptotic DE analysis for
finite-length codes will be further illustrated in
Section \ref{sec:irregular_hpc}, where we consider the parameter
optimization of so-called irregular HPCs.  We will see that by using
thresholds as an optimization criterion, performance improvements for
finite-length codes can be obtained in much the same way as for LDPC
codes.}


\section{Random Graphs and Branching Processes}
\label{sec:random_graphs}

In this section, we review the necessary background related to the
proof of Theorem \ref{th:hpc_result} in Section~\ref{sec:hpc_proof}.

\subsection{Random Graphs}

Let $\Gnp$ be the \ErdRen model (also known as the Gilbert model) of a
random graph with $n$ vertices, where each of the $\nhpc =
\binom{n}{2}$ possible edges appears with probability $p$,
independently of all other edges \cite{Erdos1959, Gilbert1959}.  A
helpful representation of this model is to consider a random,
symmetric $n \times n$ adjacency matrix $\boldsymbol{\bRV}$ with
entries $\bRV_{i,i} = 0$ and $\bRV_{i,j} (= \bRV_{j,i}) \sim
\Bern{p}$. We use $G$ to denote a random graph drawn from $\Gnp$. For
the remainder of the paper, we fix $p = c/n$.

\begin{example}
	\label{ex:vertex_degrees}
	Let $D_k = \sum_{j=1}^n \bRV_{k,j}$ be the degree of the $k$-th
	vertex. For any $k \in [n]$, $D_k \sim \Binom{n-1}{c/n}$ with
	$\E{D_k} = (n-1)c/n$. For large $n$, all degrees are approximately
	Poisson distributed with mean $c$. More precisely, let $(D_n)_{n
	\geq 1}$ be a sequence of \glspl{rv} denoting the degrees of
	randomly chosen vertices in $\Gncn$ and $D \sim \Pois{c}$. Then,
	$D_n \indistrto D$. \demo
\end{example}

The following result about the maximum vertex degree will be used in
the proof of the concentration bound \eqref{eq:hpc_concentration}. 

\begin{lemma}
	\label{lem:maximum_vertex_degree}
	Let $D_{\text{\emph{max}}} \define \max_{i \in [n]} \sum_{j=1}^n
	\bRV_{i,j}$ be the maximum degree of all vertices in the random graph
	$G$. We have
	\begin{align}	
		\Pr{D_\text{\emph{max}} \geq d_n} \leq e^{-\Omega(d_n)}, 
	\end{align}
	where $d_n$ is any function of $n$ satisfying $d_n =
	\Omega(\log(n))$. 
\end{lemma}
\begin{proof}
	The proof is standard and relies on Chernoff's inequality and the
	union bound. For completeness, a proof is given in
	Appendix~\ref{app:maximum_vertex_degree}.
\end{proof}


The random graph $G$ is completely specified by all its edges, i.e.,
by the $m$ RVs $\bRV_{i,j}$ for $1\leq j < i \leq n$. It is sometimes
more convenient to specify these RVs in a length-$m$ vector instead of
a matrix. With some abuse of notation, we also write
$\boldsymbol{\bRV} = (\bRV_1, \dots, \bRV_m)^\transpose$, asserting
that there is a one-to-one correspondence between $\bRV_k$ and
$\bRV_{i,j}$. 

\begin{example}
	Let $E= \sum_{k=1}^{m} \bRV_{k}$ be the number
	of edges in $G$.
	Then, $E \sim \Binom{\nhpc}{c/n}$ and 
	the expected number of edges grows linearly with $n$ since $\E{E} =
	\nhpc p = (n-1) c / 2 $. 
	\demo
\end{example}

\subsection{Neighborhood Exploration Process}
\label{sec:exploration_process}

An important tool to study the neighborhood of a vertex in $\Gnp$ is
the so-called exploration process which we briefly review in the
following (see, e.g., \cite[Sec.~10.4]{Alon2000},
\cite[Ch.~4]{Hofstad2014} for details). This process explores the
neighborhood in a breadth-first manner, exposing one vertex at a time.
Since we are only interested in exploring the neighborhood up to a
fixed depth, we modify the exploration compared to
\cite[Sec.~10.4]{Alon2000}, \cite[Ch.~4]{Hofstad2014} and stop the
process once all vertices in the entire neighborhood for a given depth
$\ell$ are exposed. During the exploration, a vertex can either be
active, explored, or neutral. At the beginning (time $\ttime=0$), one
vertex $v$ is active and the remaining $n-1$ vertices are neutral. At
each time $\ttime \geq 1$, we repeat the following steps.

\begin{enumerate}
	\item \RevA{Choose any of the active vertices that are closest (in
		terms of graph distance) to $v$ and denote it by $w$.
	 At
		time $\ttime=1$, choose $v$ itself.}

	\item Explore all edges $(w, w')$, where $w'$ \emph{runs through all
		active vertices}. If such an edge exists, the explored
		neighborhood is not a tree. (Apart from this fact, this step has
		no consequences for the exploration process.)

	\item Explore all edges $(w, w')$, where $w'$ \emph{runs through all
		neutral vertices}. Set $w'$ active if the edge exists. 

	\item Set $w$ explored. 

\end{enumerate}

Let $X_\ttime$ be the number of vertices that become active at time
$\ttime$ (i.e., in step 3). The number of active vertices, $A_\ttime$, and
neutral vertices, $N_\ttime$, at the end of time $\ttime$ is given by
\begin{align}
	\label{eq:exploration_implicit}
	A_{\ttime} = A_{\ttime-1} + X_\ttime - 1,  \quad N_\ttime = n -
	\ttime -
	A_\ttime,
\end{align}
with $A_0 = 1$. One can also explicitly write 
\begin{align}
	\label{eq:exploration_explicit}
	A_\ttime = S_\ttime - (\ttime-1),  \quad N_\ttime = (n - 1) -
	S_\ttime,
\end{align}
where $S_\ttime \define \sum_{i=1}^\ttime X_i$. Given $N_{\ttime-1}$, we have that
$X_\ttime \sim \Binom{N_{\ttime-1}}{p}$ because each neutral vertex can become
active at time $\ttime$ with probability $p$ \cite[p.165]{Alon2000}. 

We define the stopping time $J_\ell$ of the process $(X_\ttime)_{t\geq 1}$
to be the time when the entire depth-$\ell$ neighborhood has been
exposed.\footnote{In \cite[Sec.~10.4]{Alon2000},
\cite[Ch.~4]{Hofstad2014}, the exploration process is used to study
the connected components in $\Gnp$. In that case, the stopping time is
commonly defined as the hitting time $J \define \inf\{\ttime \in \mathbb{N}
: A_\ttime = 0\}$, i.e., the time when we run out of active vertices during
the exploration.} Formally, $J_\ell$ is recursively defined as 
\begin{align}
	\label{eq:stopping_time_definition}
	J_{\ell} = \sum_{i=1}^{J_{\ell-1}} X_i + 1 = S_{J_{\ell-1}} + 1, 
\end{align}
for $\ell \in \mathbb{N}$, where $J_0 = 0$ (i.e., $J_1 = 1$, $J_2 =
X_1 + 1$, $J_3 = \sum_{i=1}^{X_1+1} X_i + 1$, and so on). 


We further use $Z_\ell$ to denote the number of vertices at depth
$\ell$, where $Z_0 = 1$, and we let $T_\ell = \sum_{l = 0}^{\ell} Z_l$
be the total number of vertices in the entire depth-$\ell$
neighborhood. Observe that $Z_\ell = A_{J_\ell}$, i.e., the number of
vertices at depth $\ell$ corresponds to the number of active vertices
at the stopping time $J_\ell$. We also have $J_\ell = T_{\ell-1}$,
i.e., the stopping time for depth $\ell$ corresponds to the number of
all vertices up to depth $\ell-1$. 

\begin{figure}[t]
	\begin{center}
		\includegraphics{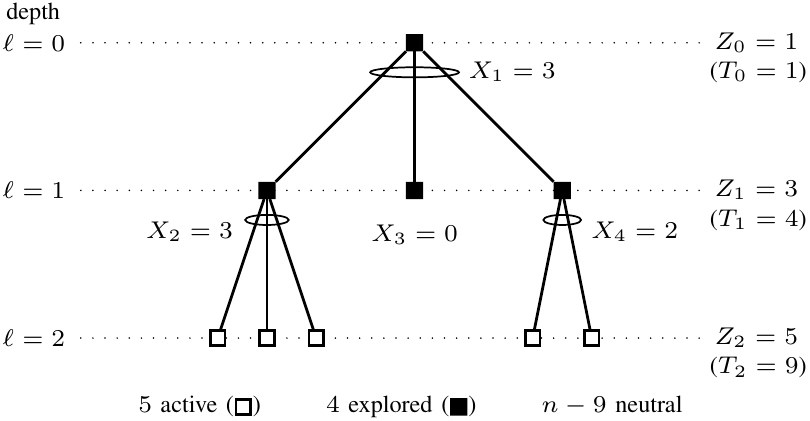}
	\end{center}
	\caption{The neighborhood of depth $\ell = 2$ after $J_2 = 4$ steps
	in the exploration corresponding to Example \ref{ex:tree_example} in
	the text.}
	\label{fig:tree_example}
\end{figure}

\begin{example}
	\label{ex:tree_example}
	Assume $\ell = 2$. An example of the neighborhood is shown in
	Fig.~\ref{fig:tree_example}. The corresponding realization of the
	stopped exploration process $(X_1, \dots, X_{J_\ell})$ is given by
	$(3, 3, 0, 2)$, where we assumed a left-to-right ordering of
	vertices. We have $J_2 = T_1 = 4$. Observe that all vertices in the
	neighborhood are exposed. However, there may still be connections
	between any of the (active) vertices at depth $2$, in which case the
	neighborhood contains cycles. \demo
\end{example}

\subsection{Branching Processes}
\label{sec:branching_process}

A (Galton--Watson) branching process with offspring distribution
$\GwXi$ is a discrete-time Markov chain $(\GwZ_\ell)_{\ell \geq 0}$
defined by \cite[Ch.~8]{Karlin1975}
	\begin{align}
		\GwZ_0 = 1 \quad \text{and} \quad \GwZ_{\ell+1} =
		\sum_{i=1}^{\GwZ_\ell}
		\GwXi_{\ell,i}, 
	\end{align}
where $(\GwXi_{\ell,i})_{\ell,i\geq 0}$ is a two-dimensional sequence
of \IID~$\No$-valued RVs with distribution $\GwXi_{\ell,i} \sim
\GwXi$. In our context, the interpretation of the process is as
follows. Start with one vertex at depth $\ell = 0$ which has a random
number of neighboring (or offspring) vertices extending to depth $1$.
Each of the vertices at depth $1$ (if there are any) has again a
random number of offspring vertices, independently of all other
vertices, and so on. $\GwZ_\ell$ is the total number of vertices at
depth $\ell$, whereas $\GwXi_{\ell,i}$ is the number of offspring
vertices of the $i$th vertex at depth $\ell$. We further define the
total number of vertices up to (and including) depth $\ell$ as
$\GwT_{\ell} = \sum_{l=0}^{\ell} \GwZ_l$. 

The exploration process in the previous subsection is closely related
to a Poisson branching process with mean $c$, i.e., the case where
$\GwXi = \Pois{c}$. The connection becomes apparent by considering the
random-walk perspective of the branching process
\cite[Sec.~3.3]{Hofstad2014}.  Here, the number of offspring vertices
is specified in a one-dimensional fashion, indexed by $\ttime$, and denoted
by $\GwX_\ttime$. The indexing is done breadth-first, in a predetermined
order, e.g., left to right. In particular, we have
\begin{align}
	\label{eq:gw_random_walk}
	\GwA_\ttime = \GwA_{\ttime-1} + \GwX_\ttime - 1
\end{align}
with $\GwA_0 = 1$, similar to \eqref{eq:exploration_implicit}. The
crucial difference with respect to the exploration process is that
$\GwX_\ttime \sim \GwXi$ for all $\ttime$. 

Similar to the exploration process, we recursively define the stopping
time for the process $(\GwX_t)_{t \geq 1}$ as $\GwJ_\ell =
\sum_{i=1}^{\GwJ_{\ell-1}} \GwX_i + 1$ with $\GwJ_0 = 0$
(cf.~\eqref{eq:stopping_time_definition}), where $\GwJ_\ell =
\GwT_{\ell-1}$. Thus, the stopped process $(\GwX_1, \cdots,
\GwX_{\GwJ_{\ell}})$ specifies the branching process up to depth
$\ell$. 

\section{Proof of Theorem 1}
\label{sec:hpc_proof}

In the following, we provide a proof of Theorem \ref{th:hpc_result}.
In Section~\ref{sec:hpc_tree_neighborhood}, we show that, with high
probability, the depth-$\ell$ neighborhood of a vertex in the residual
graph $G$ is a tree. We use this result in
Section~\ref{sec:hpc_convergence} to show the convergence of the
expected decoding outcome for an individual CN after $\ell$ iterations
to the decoding outcome when evaluated on the branching process. The
iterative decoding on the branching process (also known as \gls{de})
is analyzed in Section~\ref{sec:hpc_de}. Finally, the concentration
bound in \eqref{eq:hpc_concentration} is shown in
Section~\ref{sec:hpc_concentration}. 

The tree-like behavior and the convergence of the neighborhood in
$\Gncn$ to the Poisson branching process are certainly well-known
within the random-graph-theory literature. For example, this type of
convergence is sometimes referred to as local weak convergence, see,
e.g., \cite{Dembo2010} or \cite[Prop.~2.3.1]{Montanari2014a}. Here, we
give a simple proof based on stochastic processes and stopping times. 

\subsection{Tree-like Neighborhood}
\label{sec:hpc_tree_neighborhood}

\begin{lemma}
	\label{lem:tree_like_neighborhood}
	Let $B_G(k,\ell)$ denote the depth-$\ell$ neighborhood of the $k$-th
	vertex in $G$. Then, for any $k \in [n]$, we have
	\begin{align}
		\label{eq:tree_probability_bound}
		\Pr{B_G(k, \ell) \text{\textnormal{ is a tree}}} \geq 1 -
		\frac{\beta(c,\ell)}{ n}, 
	\end{align}
	where $\beta(c, \ell)$ depends only on $c$
	and $\ell$. 
\end{lemma}

\begin{proof}
We can use the exploration process in
Section~\ref{sec:exploration_process} to show that the total number of
potential edges that could create a cycle during the exploration
(i.e., in step 2) is given by
\begin{align}
	\label{eq:tree_proof_nl} N_\ell = \sum_{i=1}^{J_\ell} (A_{i-1}
	- 1) +\binom{Z_\ell}{2}. 
\end{align}
\RevA{In particular, at each time $\ttime$, one vertex out of the
$A_{t-1}$ active vertices is being explored. The other $A_{t-1}-1$
active vertices are known to be part of the neighborhood. Hence, any
of the $A_{t-1}-1$ potential edges to these vertices would create a
cycle. The sum in \eqref{eq:tree_proof_nl} counts all of these
potential edges up to the random stopping time $J_\ell$. 
Furthermore, at the stopping time $J_\ell$, there exist $Z_\ell$
active vertices at depth $\ell$, with $\binom{Z_\ell}{2}$ potential
edges between them, each of which creates a cycle (see, e.g.,
Fig.~\ref{fig:tree_example}).} For the neighborhood to be a tree, all
of these edges must be absent. Since any edge in the exploration will
not appear with probability $1 - c/n$, independently of all other
edges, we have
\begin{align}
	\Pr{B_G(k,\ell) \text{ is a tree}} 	&=
	\E{\left(1-\frac{c}{n}\right)^{N_\ell}} \\
	&\geq 1-
	\frac{c}{n}\E{ N_\ell }
	\label{eq:tree_proof_bound}.
\end{align}
Surely, $N_\ell$ cannot be larger than the total number of possible
edges in the neighborhood, i.e., $N_\ell \leq \binom{T_\ell}{2} \leq
T_\ell^2/2$, where we recall that $T_\ell$ is the total number of
vertices encountered. Inserting this bound into
\eqref{eq:tree_proof_bound} and using the bound
\eqref{eq:second_moment_T_ell} on $\mathbb{E}[T_{\ell}^2]$ in
Appendix~\ref{app:moments_stopping_times} (which depends only on $c$
and $\ell$) completes the proof. 
\end{proof}

\begin{remark}
	The analogous result for (regular) \gls{ldpc} code ensembles is
	given in \cite[App.~A]{Richardson2001} (see also
	\cite[Sec.~2.2]{Luby1998}). The main difference with respect to the
	proof in \cite[App.~A]{Richardson2001} (and its extension to
	irregular ensembles with bounded maximum VN and CN degree) is that
	the number of vertices in the neighborhood cannot be upper bounded
	by a constant which is independent of $n$. (In
	\cite{Richardson2001}, $n$ corresponds to the \gls{ldpc} code
	length.) 
\end{remark}

\subsection{Convergence to the Poisson Branching Process}
\label{sec:hpc_convergence}

It is well-known that the degree of a vertex in $\Gncn$ converges to a
Poisson RV with mean $c$ as $n \to \infty$ (see Example
\ref{ex:vertex_degrees}). More generally, for any finite $\ttime$, and any
$(x_1, \cdots, x_\ttime) \in \mathbb{N}_0^\ttime$, one can easily show that
(see, e.g., \cite[Sec.~4.1.2]{Hofstad2014})
\begin{align}
	\label{eq:weak_convergence_history}
	\lim_{n \to \infty} f_{X_1, \dots, X_\ttime}(x_1, \dots, x_\ttime)
	&= f_{\GwX_1}(x_1) \cdot \dots \cdot f_{\GwX_\ttime}(x_\ttime), 
\end{align}
where $\GwX_1, \dots, \GwX_\ttime$ are \IID~$\Pois{c}$. This, together with
Lemma \ref{lem:tree_like_neighborhood}, implies that the distribution
on the shape of the neighborhood (for any fixed depth) converges to a
Poisson branching process with mean $c$. To see this, note that under
the assumption that the neighborhood is tree-like, its shape is
specified by the stopped exploration process $(X_1, \dots,
X_{J_\ell})$. Each realization of $(X_1, \cdots, X_{J_\ell})$ is a
vector of some (finite) length specifying the number of offspring
vertices in the tree in a sequential manner. The set of all
realizations is thus a subset of $\mathbb{N}_0^* = \mathbb{N}_0 \cup
\mathbb{N}_0^2 \cup \mathbb{N}_0^3 \cup \cdots$.  Since
$\mathbb{N}_0^*$ is countably infinite, there exists a one-to-one
mapping between $\mathbb{N}_0^*$ and $\mathbb{N}_0$. We denote such a
mapping by $\mathcal{M} : \mathbb{N}_0^* \to \mathbb{N}_0$ and let
$\mathcal{M}^{-1}$ be its inverse. We now define new RVs $B_n =
\mathcal{M}(X_1, \cdots, X_{J_\ell})$ and $B = \mathcal{M}(\GwX_1,
\cdots, \GwX_{\GwJ_\ell})$. One can think about enumerating all
possible trees and assigning an index to each of them.  A distribution
over the shape of the trees is then equivalent to a distribution over
the indices. It is now easy to show that $B_n \indistrto B$.  For any
$b \in \mathbb{N}_0$, there exists some $\ttime$ such that
$\mathcal{M}^{-1}(b) = (x_1, \cdots, x_\ttime) \in \mathbb{N}_0^\ttime$.
Therefore, we have \newcommand{\negspace}{\!\!\!\!}
\begin{align}
	&\negspace\lim_{n \to \infty} \Pr{B_n = b} \\ 
	&\negspace = \lim_{n \to \infty} f_{J_\ell | X_1, \ldots, X_\ttime}
	(t | x_1, \cdots, x_\ttime) 
f_{X_1, \ldots, X_\ttime}(x_1, \cdots, x_\ttime)
	\label{eq:lwc_conditional2}
\\
&\negspace = f_{\GwJ_\ell | \GwX_1, \ldots, \GwX_\ttime} (t | x_1, \cdots,
x_\ttime) 
	\lim_{n \to \infty} f_{X_1, \ldots, X_\ttime}(x_1, \cdots, x_\ttime) 
	\label{eq:lwc_conditional}
	\\
	&\negspace \stackrel{\eqref{eq:weak_convergence_history}}{=}
	f_{\GwJ_\ell | \GwX_1, \ldots, \GwX_\ttime} (t | x_1, \cdots,
	x_\ttime) 
	 f_{\GwX_1}(x_1) \cdot \dots \cdot f_{\GwX_\ttime}(x_\ttime)
	 \label{eq:lwc_asymptotic} 
	 \\
	&\negspace = \Pr{B = b},
\end{align}
where, to obtain \eqref{eq:lwc_conditional} from
\eqref{eq:lwc_conditional2}, we used the fact that the conditional
distributions of the stopping times $J_\ell$ and $\GwJ_\ell$ given
$X_1, \dots, X_\ttime$ and $\GwX_1, \dots, \GwX_\ttime$, respectively,
\RevA{are both independent of $n$ and equal to 1. This is because the
realizations $x_1, \dots, x_t$ fully determine the stopping times as
$J_{\ell} = \GwJ_{\ell} = t$.}


\RevA{%
\begin{remark}
In general, the distribution $f_{J_\ell|X_1,...,X_t}$ is not
independent of $n$. From \eqref{eq:stopping_time_definition}, recall
that $J_\ell = \sum_{i=1}^{J_{\ell-1}}X_i +1$.  One may distinguish
two cases: In the first case, the realizations of $X_1, \dots, X_t$
determine $J_\ell$ and $f_{J_\ell|X_1,...,X_t}$ is an indicator
function that does not depend on $n$. In the second case, the realizations
of $X_1, \dots, X_t$ do not determine $J_\ell$ and
$f_{J_\ell|X_1,...,X_t}$ depends on $n$.  For example, let $(X_1, X_2,
X_3) = (2,2,3)$. This determines $J_{3} = 8$, i.e., $f_{J_3|X_1, X_2,
X_3}(j|2,2,3) = \mathbb{I}\{j=8\}$, where $\mathbb{I}\{\cdot \}$ is
the indicator function. However, $f_{J_3|X_1,X_2,X_3}(j|3,2,3)$
depends on $n$, since $J_3 = \sum_{i=1}^{X_1+1} X_i + 1 = X_4 + 9$.
To pass from \eqref{eq:lwc_conditional2} to
\eqref{eq:lwc_conditional}, we only encounter the first case. This is
because the random variable $B_n$ is defined as a function of the
stopped exploration process and the corresponding realizations always
determine $J_{\ell}$.  
\end{remark}
}%

A direct consequence of this result is that the expected value of a
(bounded) function applied to the neighborhood of a vertex in $\Gncn$
converges to the expected value of the same function applied to the
branching process.  In particular, recall that the RV $W\nl =
\frac{1}{n} \sum_{k=1}^n W_k\nl$ corresponds to the fraction of
component codes that declare failures after $\ell$ decoding
iterations. The indicator RV $W_k$ depends only on the shape of the
depth-$\ell$ neighborhood of the $k$-th vertex in the residual graph.
The peeling procedure can thus be written using a function
$\mathcal{D}_\ell : \mathbb{N}_0 \to \{0, 1\}$, such that 
\begin{align}
	\label{eq:exp_z_given_tree}
	\mathbb{E}[W_k\nl \,|\, \textnormal{\text{$B_G(k,\ell)$ is a tree}}] =
	\mathbb{E}[\Dec_\ell(B_n)], 
\end{align}
which, due to symmetry, is independent of $k$.  Since $B_n \indistrto
B$ and $\mathcal{D}_\ell$ is bounded, we have that
\cite[Sec.~10]{Rosenthal2006}
\begin{align}
	\lim_{n \to \infty} \mathbb{E}[\Dec_\ell(B_n)] =
	\mathbb{E}[\Dec_\ell(B)] = z^{(\ell)}, 
\end{align}
which, together with \eqref{eq:tree_probability_bound}, implies
\eqref{eq:hpc_cycle_free_bahavior}.

\begin{remark}
	It is worth mentioning that for regular \gls{ldpc} code ensembles,
	there is no notion of an asymptotic neighborhood distribution (in
	the sense of \eqref{eq:weak_convergence_history}) beyond the fact
	that cycles can be ignored. This is because the ensemble of
	computation trees for a CN (or VN) reduces to a single deterministic
	tree. 
\end{remark}

\subsection{Density Evolution}
\label{sec:hpc_de}

Once the true distribution on the neighborhood-shape has been replaced
by the branching process, the parameter of interest can be easily
computed (cf.~\cite[Sec.~2]{Pittel1996}, \cite[p.~43]{Bollobas2009a}).
In our case, the parameter of interest is the probability that a CN
declares a decoding failure after $\ell$ iterations as $n \to \infty$,
or, equivalently, the probability that the root vertex of the
branching process survives $\ell$ peeling iterations. Due to the
recursion that is inherent in the definition of the branching process,
it is not surprising that the solution is also given in terms of a
recursion. This is, of course, completely analogous to the analysis of
\gls{ldpc} code ensembles, see, e.g., the discussion in
\cite[Sec.~1]{Luby1998}. Also, similar to \gls{ldpc} codes over the
\gls{bec}, we refer to this step as \gls{de} (even though the
parameter of interest does not correspond to a density). 

Consider a Poisson branching process with mean $c$. Assume that we
have a realization of this process (i.e., a tree) up to depth $\ell$.
We wish to determine if the root vertex survives $\ell$ iterations of
the peeling procedure (and thus the CN corresponding to the root node
declares a decoding failure). One can recursively break down the answer as
follows. First, for each of the root's offspring vertices, apply
$\ell-1$ peeling iterations to the subtree that has the offspring
vertex as a root (and extends from depth $1$ to $\ell$). Then, if the
number of offspring vertices that survive this peeling is less than or
equal to $\tee$, remove the root vertex. This gives the same answer as
applying $\ell$ peeling iterations to the entire tree, since we are
simply postponing the removal decision for the root to the $\ell$-th
iteration. 

Now, in order to determine the corresponding \emph{probability} with
which the root vertex survives the peeling procedure, the crucial
observation is that the root's offspring vertices are removed
independently of each other, and with the same probability. This is a
simple consequence of the definition of the branching process and the
independence assumption between the number of offspring vertices (see
Section~\ref{sec:branching_process}). Recall that we defined the root
survival probability as $z^{(\ell)}$. Furthermore, we denote the
survival probability of the root's offspring vertices by
$x^{(\ell-1)}$. Initially, the number of offspring vertices is Poisson
distributed with mean $c$. After removing each offspring vertex
independently with probability $1-x^{(\ell-1)}$, the offspring
distribution of the root vertex follows again a Poisson distribution,
albeit with (reduced) mean $c x^{(\ell-1)}$. (This is easily seen by
using characteristic functions.) Hence, we obtain
\eqref{eq:hpc_de_z}.

Essentially the same argument can be used to determine $x^{(\ell)}$.
The only difference is that for offspring vertices we have to account
for the fact they are connected to the previous level with an edge.
Thus, they can be removed only if less than or equal to $\tee - 1$
(and not $\tee$) of their offspring vertices survive. This leads to
the recursion
\eqref{eq:hpc_de_x},
where the initial condition is given by $x^{(0)} = 1$. 

\subsection{Concentration}
\label{sec:hpc_concentration}

The concentration bound in \eqref{eq:hpc_concentration} is readily
proved by using the method of typical bounded differences
\cite{Warnke2012}. In particular, we can apply a special case of
\cite[Cor.~1.4]{Warnke2012} which is stated below (with adjusted
notation) for easier referencing. 

\begin{theorem}[\cite{Warnke2012}]
	\label{th:typical_bounded_differences_inequality}
	Let $\boldsymbol{\bRV} = (\bRV_1, \dots, \bRV_\nhpc)^\transpose$ be
	a vector of independent RVs with $\bRV_k \sim \Bern{p}$ for all $k$.
	Let $\Gamma \subseteq \{0,1\}^m$ be an event and let $f :
	\{0,1\}^\nhpc \to \mathbb{R}$ be a function that satisfies the
	following condition. There exist $\Lipschitz$ and $\Lipschitz'$ with
	$\Lipschitz \leq \Lipschitz'$ such that whenever $\boldsymbol{\bRV},
	\boldsymbol{\bRV}' \in \{0,1\}^m$ differ in only one coordinate, we
	have 
	\begin{align}
		\label{eq:typical_lipschitz_condition}
		|f(\boldsymbol{\bRV}) - f(\boldsymbol{\bRV}')| \leq  
		\begin{cases}
			\Lipschitz \quad \text{if } \boldsymbol{\bRV} \in \Gamma \\
			\Lipschitz' \quad \text{otherwise } 
		\end{cases}.
	\end{align}
	Then, for any $a \geq 0$ and any choice of $\gamma \in (0,1]$, we
	have
	\begin{equation}
		\label{eq:typical_bounded_differences_inequality}
	\begin{aligned}
		&\Pr{|f(\boldsymbol{\bRV}) - \mathbb{E}[f(\boldsymbol{\bRV})]| \geq
		a}	
		\leq m\gamma^{-1} \Pr{\boldsymbol{\bRV} \notin \Gamma}\\ 
		&+  
		\exp \left(
		- \frac{a^2}{ 2 m (1-p) p (\Lipschitz + b)^2 + 2 (\Lipschitz + b)
		a/ 3 }
		\right),
	\end{aligned}
\end{equation}
where $b = \gamma (\Lipschitz' - \Lipschitz)$.
\end{theorem}

In our context, $\boldsymbol{\bRV}$ specifies the edges in the random
graph $G$ (see Section~\ref{sec:random_graphs}). Thus, we can think
about $\vect{\theta}$ and $\vect{\theta}'$ as specifying two different
graphs $G  = G(\vect{\theta})$ and $G' = G(\vect{\theta'})$.  The
interpretation of the condition
$\eqref{eq:typical_lipschitz_condition}$ is as follows.  For any two
graphs $G, G'$ that differ in only one edge, we have $|f(G) - f(G')|
\leq \Lipschitz'$, where $f$ denotes a function applied to the graphs.
The constant $\Lipschitz'$ is often referred to as the Lipschitz
constant \cite{Warnke2012}. The event $\Gamma$ is chosen such that
changing one coordinate in $\boldsymbol{\bRV} \in \Gamma$ (i.e.,
adding or removing an edge in the graph defined by
$\boldsymbol{\bRV}$) changes the function by at most $\Lipschitz$,
where $\Lipschitz$ should be substantially smaller than $\Lipschitz'$.
The constant $\Lipschitz$ is referred to as the typical Lipschitz
constant. In this regard, the event $\Gamma$ is assumed to be a
typical event, i.e., it should occur with high probability. 

\begin{remark}
In several applications, it is possible to establish concentration
bounds based solely on suitable choices for $\Lipschitz'$. This
approach leads to the more common bounded differences inequality (also
known as McDiarmid's or Hoeffding-Azuma inequality). For example, the
concentration bound for LDPC code ensembles in
\cite[Eq.~(11)]{Richardson2001} is based on this approach. However, in
many cases (including the one considered here) the worst case changes
corresponding to $\Lipschitz'$ can be quite large, even though the
typical changes may be small. For more details, we refer the reader to
\cite{Warnke2012} and references therein.
\end{remark}

Theorem \ref{th:typical_bounded_differences_inequality} is applied as
follows. We let $f(\boldsymbol{\bRV}) = n W  = \sum_{k=1}^n W_k$.
Since $f$ is the sum of $n$ indicator RVs, we can choose $\Lipschitz'
= n$. We further let $\Gamma$ be the event that the maximum vertex
degree in $G$, denoted by $D_\text{max}$, is strictly less than
$n^\delta$ for some fixed $\delta \in (0, 1)$. For the typical
Lipschitz constant, we choose $\Lipschitz = 2 (\ell+1) n^{ \delta
\ell}$. To show that for these choices the condition
\eqref{eq:typical_lipschitz_condition} holds, we argue as follows.
First, observe that the maximum vertex degree in both $G$ and $G'$ is
at most $n^{\delta}$ since adding an edge to the graph $G$ increases
the maximum degree by at most one (and removing an edge can only
decrease the maximum degree). Consider now the maximum change in
$\sum_{k=1}^n W_k$ that can occur by adding or removing an edge
between two arbitrary vertices $i$ and $j$ under the assumption that
the maximum degree remains bounded by $n^{\delta}$. Since $W_k$
depends only on the depth-$\ell$ neighborhood of the $k$-th vertex,
such a change can only affect $W_k$ if either vertex $i$ or $j$ (or
both) are part of the neighborhood of vertex $k$. But, due to the
bounded maximum degree, vertex $i$ appears in at most
$\sum_{l=0}^{\ell} n^{\delta l} \leq (\ell+1) n^{\delta\ell}$
neighborhoods (and so does vertex $j$). Hence, the sum $\sum_{k=1}^n
W_k$ can change by at most $2 (\ell+1) n^{\delta \ell}$. 

We further choose $\gamma = n^{-1}$. Since $\Lambda' = n$, this
implies that $b \leq \gamma \Lambda' = 1$ and therefore we have
\begin{align}
	\label{eq:lambda_bound}
	(\Lambda + b) \leq (\Lambda + b)^2 \leq 4 \Lambda^2. 
\end{align}
Consider now the second term on the \gls{rhs} of
\eqref{eq:typical_bounded_differences_inequality} with $a = n \eps$
and $p = c/n$. We have 
\begin{align}
	&\exp \left(\frac{- (n \eps)^2}{ 2 m (1-c/n) c/n (\Lambda  + b)^2
	+ 2 (\Lambda + b) n \eps/ 3 } \right) \\
	\label{eq:bnd_tmp}
	&\leq \exp\left( \frac{- \eps^2 n}{ (8 c + 8 \eps/3) \Lambda^2 }
	\right) = e^{- \beta_1 n^{1- 2 \delta \ell}}
\end{align}
where the inequality in \eqref{eq:bnd_tmp} follows from $m \leq n^2$,
$1-c/n \leq 1$ and \eqref{eq:lambda_bound}, and in the last step we
used $\Lambda = 2 (\ell+1) n^{\delta \ell}$. Note that the implicitly
defined parameter $\beta_1 > 0$ depends only on $\eps$, $c$, and
$\ell$. In order to bound the first term on the \gls{rhs} of
\eqref{eq:typical_bounded_differences_inequality}, we first note that
$\Pr{\vect{\theta} \notin \Gamma} = \Pr{D_\text{max} \geq
n^\delta}$. We then have
\begin{align}
	\label{eq:bnd_tmp2}
	m\gamma^{-1} \Pr{\vect{\theta} \notin \Gamma} \leq n^3
	e^{-\beta_2 n^\delta} \leq e^{-\beta_2 n^\delta / 2},
\end{align}
where, according to Lemma \ref{lem:maximum_vertex_degree}, the first
inequality holds for some $\beta_2 > 0$ and $n$ sufficiently large.
To match the exponents in \eqref{eq:bnd_tmp} and \eqref{eq:bnd_tmp2},
we can set $\delta = (1+2\ell)^{-1}$. This proves
\eqref{eq:hpc_concentration} and completes the proof of Theorem
\ref{th:hpc_result}. 

\begin{remark}
	The above proof applies to any function of the form $f =
	\sum_{k=1}^{n} f_k$ where $f_k$ is an indicator function that
	depends only on the depth-$\ell$ neighborhood of the $k$-th vertex
	in $G$. 
\end{remark}

\section{Generalized Product Codes}
\label{sec:gpc}

In this section, we analyze a deterministic construction of
\glspl{gpc} for which the residual graph corresponds to an
inhomogeneous random graph \cite{Bollobas2007}.  The concept of
inhomogeneity naturally arises if we wish to distinguish between
different types of vertices. In our case, a type will correspond to a
particular position in the Tanner graph and a certain
erasure-correcting capability. \glspl{hpc} can be regarded as
``single-type'' or homogeneous, in the sense that all CNs (and thus
all vertices in the residual graph) behave essentially the same.  

\subsection{Code Construction}
\label{sec:gpc_construction}

Our code construction is defined in terms of three parameters $\etab$,
$\vect{\gamma}$, and $\vect{\tau}$. We denote the corresponding
\gls{gpc} by $\Gpc$, where $n$ denotes the total number of CNs in the
underlying Tanner graph.  The two parameters $\etab$ and
$\vect{\gamma}$ essentially determine the graph connectivity, where
$\etab$ is a binary, symmetric $L \times L$ matrix and $\vect{\gamma}
= (\gamma_1, \dots, \gamma_L)^\transpose$ is a probability vector of
length $L$, i.e., $\sum_{i=1}^L \gamma_i = 1$ and $\gamma_i \geq 0$.
Since GPCs have a natural representation in terms of two-dimensional
code arrays (see, e.g., Fig.~\ref{fig:code_arrays}), one may
alternatively think about $\etab$ and $\vect{\gamma}$ as specifying
the array shape. We will see in the following that different choices
for $\etab$ and $\vect{\gamma}$ recover well-known code classes.  The
parameter $\vect{\tau}$ is used to specify GPCs employing component
codes with different erasure-correcting capabilities and will be
described in more detail at the end of this subsection. 

The Tanner graph describing the \gls{gpc} $\Gpc$ is constructed as
follows. Assume that there are $L$ positions. Place $n_i \define
\gamma_i n$ \CNs at each position $i \in [L]$, where  we assume that
$n_i$ is an integer for all $i$. Then, connect each CN at position $i$
to each CN at position $j$ through a VN if and only if $\eta_{i,j} =
1$. 

In the following, we always assume that $\eta_{i,j} = 1$ for at least
one $j$ and any $i \in [L]$ so that there are no unconnected CNs.
Furthermore, we assume that the matrix $\etab$ is irreducible, so that
the Tanner graph is not composed of two (or more) disconnected graphs.

Each of the $n_i$ CNs at position $i$ has degree
\begin{align}
	\label{eq:inhomogeneous_cn_degree}
	d_{i} = \eta_{i,i}
		(n_i - 1) + \sum_{j\neq i } \eta_{i,j} n_j,
\end{align}
\RevB{where the first term in \eqref{eq:inhomogeneous_cn_degree}
arises from the convention that we cannot connect a CN to itself if
$\eta_{i,i} = 1$. Recall that the degree of a CN corresponds to the
length of the underlying component code. Thus, all component codes at
the same position have the same length. However, the component code
lengths may vary across positions depending on $\boldsymbol{\eta}$ and
$\boldsymbol{\gamma}$. }

The total number of VNs (i.e., the length of the
code) is given by
\begin{equation}
	\label{eq:inhomogeneous_number_of_edges}
\begin{aligned}
	m = \sum_{i=1}^L \eta_{i,i} \binom{n_i}{2} + \sum_{1\leq i
	< j \leq L}
	\eta_{i,j} n_i n_j \approx \frac{\vect{\gamma}^\transpose \etab
	\vect{\gamma}}{2} n^2. 
\end{aligned}
\end{equation}
In the following, we assume some fixed (and arbitrary) ordering on the
\CNs and \VNs. 

\begin{remark}
	In the light of Remark \ref{rmk:bit_assignments}, we see that the
	above construction merely specifies a Tanner graph and not a code.
	This is due to the missing assignment of the component code bit
	positions to the CN edges. Since our results do not depend on this
	assignment, it is assumed to be (arbitrarily) fixed. In the
	following examples, the assignment is implicitly specified due to an
	array description. 
\end{remark}

\begin{example}
	\label{ex:half_product_codes}
	\glspl{hpc} are recovered by considering $\etab = 1$ and
	$\vect{\gamma} = 1$. All CNs are equivalent and correspond to
	component codes of length $n-1$.  \demo
\end{example}

\begin{example}
	\label{ex:product_codes}
	Choosing $\etab = \left(\begin{smallmatrix} 0 & 1\\ 1 & 0
	\end{smallmatrix}\right)$ leads to a \gls{pc}. 
	The relative lengths of the row and column component codes can be
	adjusted through $\vect{\gamma}$, where $\vect{\gamma} = (1/2, 1/2)$
	leads to a ``square'' \gls{pc} with (uniform) component code length
	$n/2$. Note that the total number of CNs $n$ is assumed to be even
	in this case.  \demo
\end{example}

\begin{example}
	\label{ex:arbitrary_arrays}
	\RevB{Consider an arbitrarily-shaped code array of finite size which
	is composed of $n' \times n'$ blocks arranged on a grid.  In
	Fig.~\ref{fig:arbitrary_code_array}, we illustrate how to construct
	$\etab$ for such an array. This construction uses the (arbitrary)
	convention that column and row positions of the blocks are indexed
	by odd and even numbers, respectively. First, form the matrix
	$\etab'$ representing the array, where entries are $1$ if a block is
	present on the corresponding grid point and $0$ otherwise. Assume
	that $\etab'$ has size $a' \times b'$, and let $a = \max(a', b')$.
	Then, the matrix $\etab$ is of size $2a \times 2a$ (i.e., $L=2a$)
	and can be constructed by using the prescription 
	\begin{equation}
	\label{eq:eta_prescription}
	\begin{aligned}
		\eta_{2i, 2j-1} &= \eta_{i,j}', \\
		\eta_{2i-1, 2j} &= \eta_{j,i}'
	\end{aligned}
	\end{equation}
	for $i \in [a]$ and $j \in [b]$ and $\eta_{i,j} = 0$ elsewhere.  For
	example, consider a \gls{pc} where $\etab =
	\left(\begin{smallmatrix} 0 & 1\\ 1 & 0
	\end{smallmatrix}\right)$, $\boldsymbol{\gamma} = (2/5, 3/5)$, and
	$n = 20$. In this case, the row and column codes have length $8$ and
	$12$, respectively. An alternative way of describing \emph{the same} code
	is to assume that the code array is composed of $6$ blocks of
	size $4\times4$. In this case, we obtain 
	\begin{align}
		\etab' = 
		\begin{pmatrix}
			1 & 1 \\
			1 & 1 \\
			1 & 1 \\
		\end{pmatrix} \quad\text{and}\quad
		\etab = 
		\begin{pmatrix}
			0 & 1 & 0 & 1 & 0 & 1 \\
			1 & 0 & 1 & 0 & 0 & 0 \\
			0 & 1 & 0 & 1 & 0 & 1 \\
			1 & 0 & 1 & 0 & 0 & 0 \\
			0 & 0 & 0 & 0 & 0 & 0 \\
			1 & 0 & 1 & 0 & 0 & 0 
		\end{pmatrix},
	\end{align}
	where $\etab'$ is the matrix describing the array and $\etab$
	results from applying \eqref{eq:eta_prescription}. Note that CNs at
	position 5 are not connected (i.e., $\eta_{5,i}=\eta_{i,5}=0$ for
	all $i$). This is a consequence of the assumed numbering convention
	for column and row positions. After removing empty rows and columns
	from $\etab$, we can set $L=5$, $n=20$,
	and $\gamma_i = 1/5$ for all $i$ in order to obtain the same PC as
	with $\etab =
	\left(\begin{smallmatrix} 0 & 1\\ 1 & 0
	\end{smallmatrix}\right)$, $\boldsymbol{\gamma} = (2/5, 3/5)$, and
	$n = 20$.} \demo
\end{example}

\begin{figure}[t]
	\begin{center}
		\includegraphics{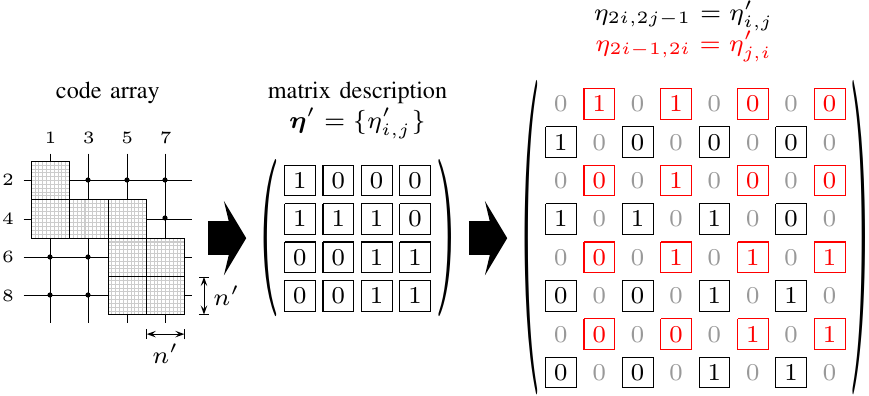}
	\end{center}
	\caption{Construction of $\etab$ for an arbitrary code array composed
	of \RevB{$n' \times n'$} blocks that are arranged on a grid. Red elements in $\etab$ are
	inserted such that $\etab$ is symmetric. With this construction, even
	(odd) positions in $\etab$ correspond to row (column) codes. }
	\label{fig:arbitrary_code_array}
\end{figure}

\begin{figure}[t]
	\centering
	\subfloat[staircase code]{\includegraphics{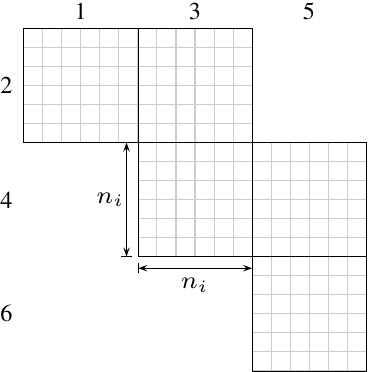}}
	\qquad
	\subfloat[block-wise braided code]{\includegraphics{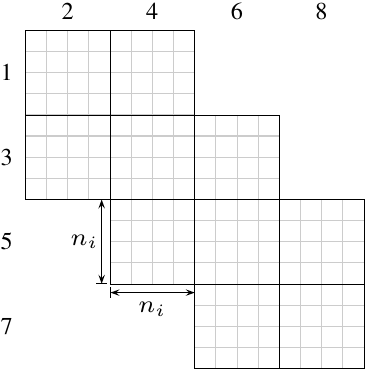}}
	\caption{Examples of code arrays for (a) staircase codes (see
	Example \ref{ex:staircase_codes}) and (b) block-wise braided
	codes (see Example \ref{ex:braided_codes}).  }
	\label{fig:code_arrays}
\end{figure}

\begin{example}
	\label{ex:staircase_codes}
	For a fixed $L \geq 2$, the matrix $\etab$ describing a staircase
	code \cite{Smith2012a} has entries $\eta_{i, i+1} = \eta_{i+1,i} =
	1$ for $i \in [L-1]$ and zeros elsewhere. The distribution
	$\vect{\gamma}$ is uniform, i.e., $\gamma_i = 1/L$ for all $i \in
	[L]$. For example, the staircase code corresponding to the code
	array shown in Fig.~\ref{fig:code_arrays}(a), where $L = 6$ and $n =
	36$ (i.e., $n_i = 6$), is defined by
	\begin{align}
			\label{eq:eta_staircase_code}
	\etab = 
		\begin{pmatrix}
			0 & 1 & 0 & 0 & 0 & 0\\
			1 & 0 & 1 & 0 & 0 & 0\\
			0 & 1 & 0 & 1 & 0 & 0\\
			0 & 0 & 1 & 0 & 1 & 0\\
			0 & 0 & 0 & 1 & 0 & 1\\
			0 & 0 & 0 & 0 & 1 & 0\\
		\end{pmatrix} ,
	\end{align}
	and $\gamma_i = 1/6$. The CNs at all positions have the same degree
	$2 n \gamma_i = 12$, except for positions $1$ and $L$, where the degrees
	are $n \gamma_i = 6$.  \demo
\end{example}

\begin{example}
	\label{ex:braided_codes}
	For even $L\geq 4$, the matrix $\etab$ for a particular instance of a
	block-wise braided code has entries $\eta_{i,i+1} = \eta_{i+1,
	i} = 1$ for $i \in [L-1]$, $\eta_{2i-1, 2i+2} = \eta_{2i+2, 2i-1} =
	1$ for $i \in [L/2-1]$, and zeros elsewhere. For example, we have 
	\begin{align}
		\label{eq:eta_braided_code}
		\etab = \begin{pmatrix}
			0 & 1 & 0 & 1 & 0 & 0 & 0 & 0\\
			1 & 0 & 1 & 0 & 0 & 0 & 0 & 0\\
			0 & 1 & 0 & 1 & 0 & 1 & 0 & 0\\
			1 & 0 & 1 & 0 & 1 & 0 & 0 & 0\\
			0 & 0 & 0 & 1 & 0 & 1 & 0 & 1\\
			0 & 0 & 1 & 0 & 1 & 0 & 1 & 0\\
			0 & 0 & 0 & 0 & 0 & 1 & 0 & 1\\
			0 & 0 & 0 & 0 & 1 & 0 & 1 & 0\\
		\end{pmatrix} 
	\end{align}
	for $L = 8$.  The corresponding code array is shown in
	Fig.~\ref{fig:code_arrays}(b), where $n = 32$ and $\vect{\gamma}$ is
	uniform.  In general, the construction of a block-wise braided code
	is based on so-called \glspl{mbp}. An \gls{mbp} with multiplicity
	$k$ is an $N \times N$ matrix with $k$ ones in each row and column
	\cite[Def.~2.1]{Feltstrom2009}.  Given a component code of length
	$\nham$ and dimension $\kham$, the diagonal and off-diagonal array
	blocks in Fig.~\ref{fig:code_arrays}(b) correspond to \glspl{mbp}
	with respective multiplicities $2\kham-\nham$ and $\nham - \kham$,
	where $N \geq \min(2\kham-\nham, \nham - \kham)$. However, this
	definition is unnecessarily narrow for our purposes in the sense
	that the multiplicities of the \glspl{mbp} are linked to the
	dimension of the component code. For example, for the array shown in
	Fig.~\ref{fig:code_arrays}(b) (where $N=4$ and $\nbch = 12$), it
	would be required that each component code has dimension $\kbch = 8$
	in order to comply with the definition in \cite{Feltstrom2009}.
	Here, we simply lift the constraint that the multiplicities of the
	\glspl{mbp} are linked to the component code dimension. The only
	requirement for the considered \gls{gpc} construction is that the
	\glspl{mbp} have a block-wise structure themselves, see
	Fig.~\ref{fig:braided_code_array_original} for an example. Note that
	$\etab$ can be found by following the steps in Example
	\ref{ex:arbitrary_arrays}.  \demo
\end{example}

\begin{figure}[t]
	\begin{center}
		\includegraphics{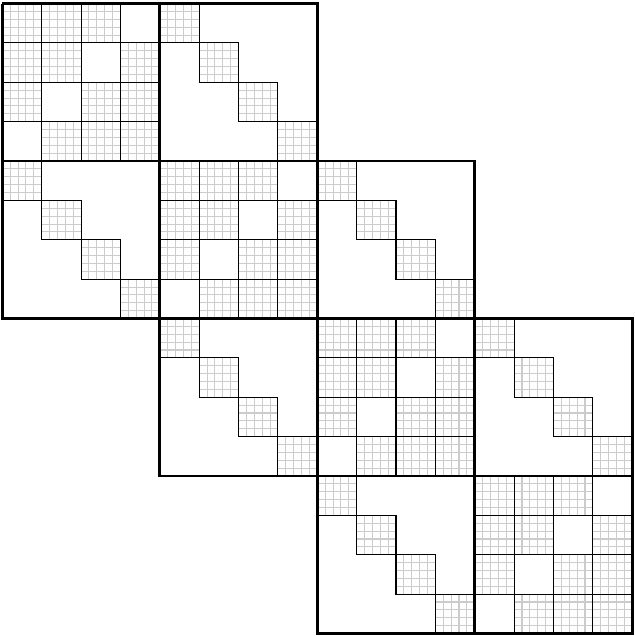}
	\end{center}
	\caption{A block-wise braided code where the \glspl{mbp} have
	a block-wise structure. We have $N = 20$, $n = 160$, and the
	multiplicities for the diagonal and off-diagonal \glspl{mbp} are
	$15$ and $5$, respectively. }
	\label{fig:braided_code_array_original}
\end{figure}

\begin{remark}
	Both staircase and braided codes were originally introduced as
	convolutional-like codes with conceptually infinite length, i.e., $L
	= \infty$. It then becomes customary to employ a sliding-window
	decoder whose analysis is discussed in
	Section~\ref{sec:modified_decoding}. We also remark that it is
	straightforward to extend the above construction of $\etab$ and
	$\vect{\gamma}$ for staircase and braided codes to their natural
	tail-biting versions (see, e.g., \cite{Justesen2011}).
	\NoRev{Staircase and braided codes can be classified as instances of
	(deterministic) spatially-coupled PCs, which are discussed in more
	detail in Section \ref{sec:spatially_coupled}.}
\end{remark}

Up to this point, the \gls{gpc} construction for a given $\etab$,
$\vect{\gamma}$, and $n$ specifies the lengths of the component codes
via \eqref{eq:inhomogeneous_cn_degree}. We proceed by assigning
different erasure-correcting capabilities to the component codes
corresponding to CNs at different positions. To that end, for $i \in
[L]$, let $\vect{\tau}(i) = (\tau_1(i), \dots,
\tau_{\tmax}(i))^\transpose$ be a probability vector of length
$\tmax$, where $\tau_{\tee}(i)$ denotes the fraction of CNs at
position $i$ (out of $n_i$ total CNs) which can correct $\tee$
erasures and $\tmax$ is the maximum erasure-correcting capability.
With some abuse of notation, the collection of these distributions for
all positions is denoted by $\vect{\tau} = (\vect{\tau}(i))_{i=1}^L$.
The assignment can be done either deterministically, assuming that
$\tau_\tee(i) n_i$ is an integer for all $i \in [L]$ and $\tee \in
[\tmax]$, or independently at random according to the distribution
$\vect{\tau}(i)$ for each position. 

\begin{example}
	Consider a \gls{pc} where the row and column codes have the same
	length but different erasure-correcting capabilities $\tee$ and
	$\tee'$, respectively. We have $\etab = \left(\begin{smallmatrix} 0 &
		1\\ 1 & 0
	\end{smallmatrix}\right)$, $\vect{\gamma} = (1/2, 1/2)$, and additionally
	$\tau_{\tee}(1) = 1$ and $\tau_{\tee'}(2) = 1$. More generally, the
	erasure-correcting capabilities may also vary across the row (and
	column) codes leading to irregular \glspl{pc} \cite{Hirasawa1984, Alipour2012}.
	\demo
\end{example}

\begin{example}
	Staircase codes with component code mixtures were suggested (but not
	further investigated) in \cite[Sec.~4.4.1]{Smith2011}. The case
	described in \cite[Sec.~4.4.1]{Smith2011} corresponds to a fixed
	choice of $\vect{\tau}(i)$ which is independent of $i$. \RevA{A code
	ensemble that is structurally related to staircase codes based on
	component code mixtures was analyzed in \cite[Sec.~III]{Zhang2015}.}
	\demo
\end{example}

\subsection{Inhomogeneous Random Graphs}
\label{sec:inhomogeneous_random_graphs}

Assume that a codeword of $\Gpc$ is transmitted over the \gls{bec}
with erasure probability $p = c/n$, for $c > 0$. Recall that the
residual graph is obtained by removing known VNs and collapsing erased
VNs into edges. We now illustrate how the ensemble of residual graphs
for $\Gpc$ is related to the inhomogeneous random graph model.

\RevB{
\begin{remark}
	The parameter $c$ can again be interpreted as an ``effective''
	channel quality. Unlike for HPCs, its operational meaning does not
	necessarily correspond to the expected number of initial erasures
	per component code constraint (see Section \ref{sec:hpc_bec}).
	However, the construction of $\Gpc$ can be slightly modified by
	introducing a parameter $a > 0$ that scales the total number of CNs
	from $n$ to $a n$. In this case, $a$ can be chosen such that $c$
	corresponds again to the average number of erasures per component
	code constraint for large $n$. This modified construction is
	discussed in Section \ref{sec:discussion_thresholds} below.
\end{remark}
}

In \cite{Bollobas2007}, inhomogeneous random graphs are specified by a
vertex space $\mathcal{V}$ and a \RevB{kernel $\boldsymbol{\kappa}$}.
Here, we consider only the finite-type case, see
\cite[Ex.~4.3]{Bollobas2007}. In this case, the number of different
vertex types is denoted by $r$ and the vertex space $\mathcal{V}$ is a
triple $(\mathcal{S}, \mu, (\boldsymbol{y}_n)_{n\geq 1})$, where
$\mathcal{S} = [r]$ is the so-called type space, $\mu : \mathcal{S}
\to [0, 1]$ is a probability measure on $\mathcal{S}$, and
$\boldsymbol{y}_n = (y_1^{(n)}, y_2^{(n)}, \dots, y_n^{(n)})$ is a
deterministic or random sequence of points in $\mathcal{S}$ such that
for each $i \in \mathcal{S}$, we have
\begin{align}
	\label{eq:fraction_of_vertices}
	\frac{|\{k : y_k^{(n)} = i\}|}{n} \inprobto \mu(i) 
\end{align}
as $n \to \infty$. For a finite number of vertex types, the \RevB{kernel
$\boldsymbol{\kappa}$} is a symmetric $r \times r$ matrix, where
entries are denoted by $\kappa_{i,j}$. For a fixed $n > \max_{i,j}
\kappa_{i,j}$, the inhomogeneous random graph $\Gnk$ is defined as
follows. The graph has $n$ vertices where the type of vertex $i$ is
given by $y^{(n)}_i$.  An edge between vertex $i$ and $j$ exists with
\RevB{probability $n^{-1}\kappa_{y^{(n)}_i,y^{(n)}_j}$}, independently of all
other edges. 

\begin{remark}
	Even though we use \cite{Bollobas2007} as our main reference,
	finite-type inhomogeneous random graphs (and their relation to
	multi-type branching processes) were first introduced and studied in
	\cite{Soderberg2002}. See also the discussion in
	\cite[p.~31]{Bollobas2009a}. 
\end{remark}

For the code $\Gpc$, the residual graph is  an instance of an
inhomogeneous random graph with a finite number of types, as defined
above. In particular, there are $r = L \tmax$ different types in
total, i.e., we have $\mathcal{S} = [L \tmax]$. In our case, it is
more convenient to specify the type of a vertex by a pair $(i, \tee)$,
where $i \in [L]$ corresponds to the position in the Tanner graph and
$\tee \in [\tmax]$ corresponds to the erasure-correcting capability.
In the construction of the sequence $\boldsymbol{x}_n$, the assignment
of the type corresponding to the position is always deterministic. For
the type corresponding to the erasure-correcting capability, we have
the freedom to do the assignment deterministically or uniformly at
random. In both cases, the fraction of vertices of type $(i, \tee)$ is
asymptotically given by $\gamma_i \tau_\tee(i)$. This specifies the
probability measure $\mu$ through the condition
\eqref{eq:fraction_of_vertices}. (For the random assignment, the condition
\eqref{eq:fraction_of_vertices} holds due to the weak law of large
numbers.) The kernel $\boldsymbol{\kappa}$ is obtained from $\etab$ by
replacing each $0$ entry with the all-zero matrix of size $\tmax
\times \tmax$ and each $1$ entry with a $\tmax \times \tmax$ matrix
where all entries are equal to $c$.

\begin{remark}
	The inhomogeneous random graph model in \cite{Bollobas2007} is much
	more general than the finite-type case described above. In
	particular, $\mathcal{S}$ can be a separable metric space and
	$\kappa$ a symmetric non-negative (Borel) measurable function on
	$\mathcal{S} \times \mathcal{S}$. This more general framework could
	be used for example to obtain the \gls{de} equations for so-called
	tightly-braided codes \cite{Feltstrom2009, Jian2014}. However, in
	that case the analysis does not admit a characterization in terms of
	a finite number of types. In particular, the \gls{de} equations are
	given in terms of integrals and solving the equations may then
	require the application of numerical integration techniques.
\end{remark}

\begin{figure}[t]
	\begin{center}
		\includegraphics{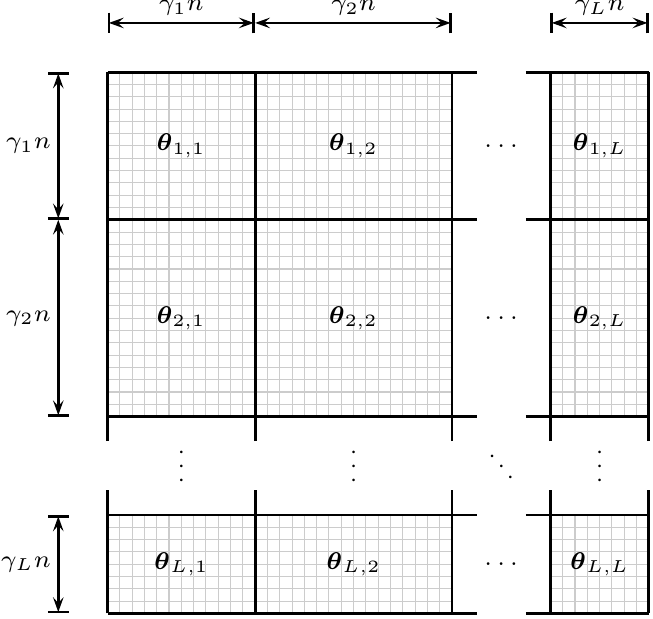}
	\end{center}
	\caption{The structure of the random, symmetric adjacency matrix
	$\boldsymbol{\bRV}$.}
	\label{fig:adjacency_matrix}
\end{figure}

Similar to the (homogeneous) random graph $\Gnp$, one may use an
alternative representation in terms of a random, symmetric $n \times
n$ adjacency matrix $\boldsymbol{\bRV}$ with zeros on the diagonal.
The structure of this matrix is shown in
Fig.~\ref{fig:adjacency_matrix}. The matrix is composed of submatrices
$\boldsymbol{\bRV}_{i,j}$ of size $n_i \times n_j$. The submatrix
$\boldsymbol{\bRV}_{i,j}$ is zero if $\eta_{i,j} = 0$ and it consists
of \IID $\Bern{p}$ RVs if $\eta_{i,j} = 1$ (with the constraint that
the matrix $\boldsymbol{\bRV}$ is symmetric and all diagonal elements
are zero). The inhomogeneous random graph is thus specified by $m$
Bernoulli RVs, where $m$ is defined in
\eqref{eq:inhomogeneous_number_of_edges}. 

From the matrix representation, it can be seen that the degree of a
vertex at position $k$ is distributed according to $\Binom{d_k}{c/n}$,
where $d_k$ is defined in \eqref{eq:inhomogeneous_cn_degree}.
Moreover, all vertices follow a Poisson distribution as $n \to
\infty$, where the mean for vertices at position $k$ is given by
\begin{align}
	\label{eq:gpc_poisson_mean}
	\lim_{n \to \infty} d_k \frac{c}{n} = c \sum_{j=1}^L \gamma_j
	\eta_{k, j}. 
\end{align}

\subsection{Iterative Decoding}
\label{sec:gpc_decoding}

After transmission over the \gls{bec}, we apply the same iterative
decoding as described in Section~\ref{sec:hpc_decoding} for the
\gls{hpc}. The only difference is that each component code is assumed
to correct all erasures up to its erasure-correcting capability. In
the corresponding iterative peeling procedure for the residual graph,
one removes vertices of degree at most $\tee$, where $\tee$ is the
erasure-correcting capability of the corresponding CN. The
erasure-correcting capability may now be different for different
vertices depending on their type. 

\subsection{Asymptotic Performance}
\label{sec:gpc_performance}

For a fixed number of decoding iterations $\ell$, we wish to
characterize the asymptotic performance as $n \to \infty$. The crucial
observation is that the distribution on the neighborhood-shape of a
randomly chosen vertex in the residual graph converges asymptotically
to a multi-type branching process \cite[Remark 2.13]{Bollobas2007}. In
our case, the multi-type branching process is defined in terms of the
code parameters $\etab$, $\vect{\gamma}$, and $\vect{\tau}$. It
generalizes the (single-type) branching process described in
Section~\ref{sec:branching_process} as follows. The process starts with
one vertex at depth $0$ which has random type $(i, \tee)$ with
probability $\gamma_i \tau_{\tee}(i)$. This vertex has neighboring (or
offspring) vertices of possibly different types that extend to depth
$1$, each of which has again neighboring vertices that extend to the
next depth, and so on. For a vertex with type $(i, \tee)$, the number
of offspring vertices with type $(j, \tee')$ is Poisson distributed
with mean $c \eta_{i,j} \gamma_j \tau_{\tee'}(j)$, independently of
the number of offspring vertices of other types.  Since the sum of
independent Poisson RVs is again Poisson distributed, we have that the
total number of offspring vertices of a vertex with type $(i, \tee)$
is Poisson distributed with mean 
\begin{align}
	\sum_{j=1}^L \sum_{\tee' = 1}^{\tmax} c \gamma_j \eta_{i,j}
	\tau_{\tee'}(j) = c \sum_{j=1}^{L} \gamma_j \eta_{i,j}, 
\end{align}
independently of $\tee$ (cf.~\eqref{eq:gpc_poisson_mean}). The above
multi-type branching process is denoted by $\mathfrak{X}$. We further
use $\mathfrak{X}({i, \tee})$ to denote the process which starts with
a root vertex that has the specific type $(i, \tee)$. 

Let $z^{(\ell)}$ be the probability that the root vertex of the first
$\ell$ generations of $\mathfrak{X}$ survives the $\ell$ iterations of
the peeling procedure. This probability is evaluated explicitly in
Section~\ref{sec:gpc_de} in terms of the code parameters $\etab$,
$\vect{\gamma}$, and $\vect{\tau}$. The main result is as follows.

\begin{theorem}
	\label{th:gpc_result}
	Let $W_k\nl$, $k \in [n]$, be the indicator \gls{rv} for the event
	that the $k$-th component code of $\Gpc$ declares a decoding failure after $\ell$
	iterations of decoding and define $W\nl = \frac{1}{n} \sum_{k=1}^n
	W_k\nl$. Then, we have
\begin{align}
	\label{eq:gpc_cycle_free_bahavior}
	\lim_{n \to \infty} \mathbb{E}[W\nl] =  z^{(\ell)}.
\end{align}
Furthermore, for any $\eps > 0$, there exist $\delta > 0$, $\beta>0$,
and $n_0 \in \mathbb{N}$ such that for all $n > n_0$ we have
\begin{align}
	\label{eq:gpc_concentration}
	\Pr{|W\nl - \mathbb{E}[W\nl]| > \eps} \leq 
	e^{-\beta n^\delta}.
\end{align}
\end{theorem}

\begin{proof}
In order to prove \eqref{eq:gpc_cycle_free_bahavior}, we apply
\cite[Th.~11.6]{Bollobas2007}. First, recall the following definition
from \cite[p.~74]{Bollobas2007}. Let $f(v, G)$ be a function defined
on a pair $(v, G)$, where $G$ is a graph composed of vertices with
different types and $v$ is a distinguished vertex of $G$, called the
root. The function $f$ is an $\ell$-neighborhood function if it is
invariant under type-preserving rooted-graph isomorphisms and depends
only on the neighborhood of the vertex $v$ up to depth $\ell$. The RV
$W_k$ depends only on the depth-$\ell$ neighborhood of the $k$-th
vertex in the residual graph of $\Gpc$. Furthermore, the peeling
outcome for the $k$-th vertex is invariant under isomorphisms as long
as they preserve the vertex type. Hence, the RV $W_k$ can be expressed
in terms of an $\ell$-neighborhood function $\mathcal{D}_\ell$ as $W_k
= \mathcal{D}_\ell(k, G)$. That is, the function $\mathcal{D}_\ell$
evaluates the peeling procedure on the depth-$\ell$ neighborhood of a
vertex and thus determines if the corresponding component code
declares a decoding failure after $\ell$ iterations. To apply
\cite[Th.~11.6]{Bollobas2007}, we need to check that we have $\sup_n
\E{\mathcal{D}_\ell(k, G)^4} < \infty$. This is true since since
$D_\ell$ maps to $\{0,1\}$. We then have
\cite[Eq.~(11.4)]{Bollobas2007}
\begin{align}
	\lim_{n \to \infty} \E{W} = \E{\mathcal{D}_\ell(\mathfrak{X})}, 
\end{align}
where $\mathcal{D}_\ell(\mathfrak{X})$ is defined by evaluating
$\mathcal{D}_\ell$ on the branching process $\mathfrak{X}$ up to depth
$\ell$, taking the initial vertex as the root. Therefore,
$\E{\mathcal{D}_\ell(\mathfrak{X})} = z^{(\ell)}$. This result
generalizes the convergence result \eqref{eq:hpc_cycle_free_bahavior}
for \glspl{hpc} (i.e., $\Gncn$) shown in
Sections~\ref{sec:hpc_tree_neighborhood} and \ref{sec:hpc_convergence}.
The proof in \cite{Bollobas2007} relies on a stochastic coupling of
the branching process $\mathfrak{X}$ and the neighborhood exploration
process for $\Gnk$ (which generalizes the exploration process
described in Section~\ref{sec:exploration_process} to handle different
vertex types), see \cite[Lem.~11.4]{Bollobas2007} for details. 

The proof of \eqref{eq:gpc_concentration} follows along the same lines
as the proof for the homogeneous case in
Section~\ref{sec:hpc_concentration}, using again the typical bounded
differences inequality. The bound on the maximum vertex degree for
$\Gncn$ in Lemma \ref{lem:maximum_vertex_degree} applies without
change also to the inhomogeneous random graph $\Gnk$. The only
difference in the proof in Appendix \ref{app:maximum_vertex_degree} is
that the equality in \eqref{eq:maximum_vertex_degree_independence}
becomes an inequality. The choice of the high-probability event
$\Gamma$ and the typical Lipschitz constant is then the same as
described in Section~\ref{sec:hpc_concentration}. 
\end{proof}

\subsection{Density Evolution}
\label{sec:gpc_de}

In order to compute $z^{(\ell)}$, we proceed in a similar fashion as
described in Section~\ref{sec:hpc_de} and break down the computation in a
recursive fashion. First, note that from the definition of
$\mathfrak{X}$ and $\mathfrak{X}(i, \tee)$, we have
\begin{align}
	\label{eq:de_gpc_ver1a}
	z^{(\ell)} = \E{\mathcal{D}_\ell(\mathfrak{X})} = \sum_{i=1}^{L}
	 \sum_{\tee = 1}^{\tmax} \gamma_i \tau_\tee(i) z_{i,
	\tee}^{(\ell)},
\end{align}
where
\begin{align}
	z_{i,\tee}^{(\ell)} = \E{\mathcal{D}_\ell(\mathfrak{X}({i, \tee}))}
\end{align}
is the probability that the root vertex of the first $\ell$
generations of the branching process $\mathfrak{X}(i, \tee)$ survives
the peeling procedure. We claim that 
\begin{align}
	\label{eq:de_gpc_ver1b}
	z_{i,\tee}^{(\ell)} = \Psi_{\geq \tee + 1} \left( c
	\sum_{j=1}^{L}  \sum_{\tee'=1}^{\tmax}
	\eta_{i,j} \gamma_j\tau_{\tee'}(j)
	x_{j,\tee'}^{(\ell-1)}
	\right), 
\end{align}
where $x_{i,\tee}^{(\ell)}$ is recursively given by 
\begin{align}
	\label{eq:de_gpc_ver1c}
	x_{i,\tee}^{(\ell)} = \Psi_{\geq \tee} \left( c
	\sum_{j=1}^{L} \sum_{\tee'=1}^{\tmax}
	\eta_{i,j} \gamma_j \tau_{\tee'}(j)
	x_{j,\tee'}^{(\ell-1)}
	\right), 
\end{align}
with $x_{i, \tee}^{(0)} = 1$. The argument is the same as described in
Section~\ref{sec:hpc_de}. In particular, to determine the survival of the
root of $\mathfrak{X}(i, \tee)$ after $\ell$ peeling iterations, first
determine if each offspring vertex gets removed by applying $\ell-1$
peeling iterations to the corresponding subtree. Then, make a decision
based on the number of surviving offspring vertices. Again, one finds
that offspring vertices survive independently of each other, however,
the survival probability now depends on the vertex type. In
particular, in \eqref{eq:de_gpc_ver1b}, the quantity
$x_{i,\tee}^{(\ell-1)}$ is the probability that a type-$(i, \tee)$
offspring of the root vertex survives the $\ell-1$ peeling iterations
applied to its subtree. The argument of the function $\Psi_{\geq
\tee+1}$ in \eqref{eq:de_gpc_ver1b} is the mean number of surviving
offspring vertices, which, again, is easily found to be Poisson
distributed. Essentially the same arguments can be applied to find
\eqref{eq:de_gpc_ver1c} by taking into account the connecting edge of
each offspring vertex to the previous level of the tree. 

Using the substitution $x_i^{(\ell)} = \sum_{\tee = 1}^{\tmax}
\tau_\tee(i) x_{i,\tee}^{(\ell)}$, it is often more convenient to
express $z^{(\ell)}$ in terms of
\RevA{%
\begin{align}
	\label{eq:z_ell_definition}
	z^{(\ell)} = 
	\sum_{i=1}^{L} \gamma_i 
	\sum_{\tee=1}^{\tmax} \tau_\tee(i) \Psi_{\geq \tee+1}
	\left(
		c \sum_{j=1}^{L} \eta_{i,j} \gamma_j x_j^{(\ell-1)}
	\right), 
\end{align}
}%
where
\begin{align}
	\label{eq:de_gpc}
	x_i^{(\ell)} = \sum_{\tee=1}^{\tmax} \tau_\tee(i) \Psi_{\geq \tee}
	\left(
		c \sum_{j=1}^{L} \eta_{i,j} \gamma_j x_j^{(\ell-1)}
	\right)
\end{align}
with $x_i^{(0)} = 1$ for all $i \in [L]$. 

\section{Discussion}

Before considering a direct application of the obtained DE equations
in the next section, we briefly discuss some relevant topics regarding
their general application. 

\subsection{Thresholds and Code Comparisons}
\label{sec:discussion_thresholds}

The decoding threshold for $\Gpc$ can be defined in terms of the effective
channel quality as
\begin{align}
	\label{eq:gpc_threshold}
	\cthr \define \sup\{ c > 0 \,|\, \lim_{\ell \to \infty} z^{(\ell)} = 0
	\}.
\end{align}
Recall that in the code construction in
Section~\ref{sec:gpc_construction}, $\vect{\gamma}$ is assumed to be a
distribution, i.e., we have $\sum_{i=1}^L \gamma_i = 1$. This
assumption turns out to be convenient in the formulation and proof of
Theorem \ref{th:gpc_result} since it ensures that the number of CNs in
the Tanner graph is always given by $n$. However, when comparing the
performance of different \glspl{gpc} (for example in terms of
thresholds computed via \eqref{eq:gpc_threshold}), it is more
appropriate to lift this assumption and replace $\gamma_i$ by a
rescaled version $a \gamma_i$ for all $i$ and some constant $a$. This
simply corresponds to scaling the total number of CNs to $an$. 

A reasonable scaling to compare different codes is to choose $a$ such
that the effective channel quality $c$ can be interpreted
asymptotically as the average number of initial erasures in each
component code, similar to \glspl{hpc} in
Section~\ref{sec:hpc_decoding}.  Since each component code at position
$i$ initially contains $a n_{\mathcal{C}, i} c/n$ erasures, by
averaging over all positions we obtain 
\begin{align}
	&\lim_{n \to \infty} a \frac{c}{n} \sum_{i=1}^L \gamma_i \left( \sum_{j \neq i} \gamma_j
	n \eta_{i,j} + \eta_{i,i} (\gamma_i n - 1) \right)\\
	&=
	a c \sum_{i=1}^L \gamma_i \sum_{j = 1}^{L} \gamma_j
	\eta_{i,j} = a c \vect{\gamma}^\transpose \etab \vect{\gamma}.
	\label{eq:scaling_tmp}
\end{align}
Setting \eqref{eq:scaling_tmp} equal to $c$ leads to 
\begin{align}
	\label{eq:scaling}
	a = \frac{1}{\vect{\gamma}^\transpose \etab \vect{\gamma}}. 
\end{align}

\begin{example}
	For staircase codes (see Example \ref{ex:staircase_codes}), we
	obtain $a = (2L-2)/L^2$. For large $L$, $a \approx 2/L$ so that $a
	\gamma_i \approx 1/2$. \demo
\end{example}

\subsection{Upper Bound on the Decoding Threshold}
\label{sec:upper_bound}

An upper bound on the decoding threshold for $\Gpc$ can be given as
follows, see \cite[Sec.~VI-A]{Jian2012}. Assume for a moment that all
component codes can correct up to $\tee$ erasures. The best one can
hope for in this case is that each component code corrects exactly
$\tee$ erasures. That is, in total at most $a \tee n$ erasures can be
corrected, where $a$ is assumed to be defined as in
\eqref{eq:scaling}.  \RevB{Normalizing by the code length $m$ (see
\eqref{eq:inhomogeneous_number_of_edges})} gives a maximum erasure
probability of $p \leq a \tee n / m$ or, in terms of the effective
channel quality $c \leq a \tee n^2 / m$. Using
\eqref{eq:inhomogeneous_number_of_edges} as $n \to \infty$, we obtain
$c \leq 2 \tee$ as a necessary condition for successful decoding. This
reasoning extends naturally also to the case where we allow for a
mixture of erasure-correcting capabilities. In this case, one finds
that $c \leq 2 \bar{\tee}$, where
\begin{align}
	\label{eq:average_tee}
	\bar{\tee} = \sum_{i=1}^{L} \gamma_i \sum_{\tee = 1}^{\tmax}
	\tau_\tee(i) \tee
\end{align}
is the mean erasure-correcting capability. This bound is used for
example as a reference in the code optimization discussed in
Section~\ref{sec:irregular_hpc}. 

\RevA{\begin{remark} A similar discussion can be found in
	\cite{Zhang2015}, where the authors refer to the resulting threshold
	bound (see \cite[Def.~1]{Zhang2015}) as the
	``weight-pulling'' threshold.
\end{remark}}

\subsection{Modified Decoding Schedules}
\label{sec:modified_decoding}

We now discuss decoding algorithms that differ from the one described
in Section~\ref{sec:hpc_decoding} and Section~\ref{sec:gpc_decoding} in
terms of scheduling. For example, for conventional \glspl{pc}, one
typically iterates between the component decoders for the row and
column codes. Another example is the decoding of convolutional-like
\GPCs, such as the ones described in Examples~\ref{ex:staircase_codes}
and \ref{ex:braided_codes}. For these codes, $L$ is typically assumed
to be very large and it becomes customary to employ a sliding-window
decoder. Such a decoder does not require knowledge of the entire
received code array in order to start decoding. The decoder instead
only operates on a subset of the array within a so-called window
configuration. After a predetermined number of iterations, this subset
changes and the window ``slides'' to the next position. 

More generally, assume that we wish to apply a different decoding
schedule to $\Gpc$. To that end, let $\mathcal{A}^{(l)} \subseteq [L]$
for $l \in [\ell]$ be a subset of the $L$ CN positions. We interpret
$\mathcal{A}^{(l)}$ as \emph{active} positions and the complementary
set $[L] \setminus \mathcal{A}^{(l)}$ as \emph{inactive} positions in
iteration $l$. The decoding is modified as follows. In iteration $l$,
one only executes the \gls{bdd} corresponding to \CNs at active
positions, i.e., positions that are contained in the set
$\mathcal{A}^{(l)}$. \CNs at inactive positions are assumed to be
frozen, in the sense that they do not contribute to the decoding
process. In the peeling procedure, vertices at inactive positions are
simply ignored during iteration $l$. 

In order to check if Theorem \ref{th:gpc_result} remains valid for a
modified decoding schedule, we adopt the convention that frozen CNs
continue to declare a decoding failure if they declared a failure in
the iteration in which they were last active. Moreover, we assume that
each CN position belongs to the set of active positions at least once
during the decoding, i.e., we assume that $\bigcup_{l = 1}^{\ell}
\mathcal{A}^{(l)} = [L]$ (otherwise $W_k$ in Theorem
\ref{th:gpc_result} is not defined for CNs that were never activated).
Using these assumptions, it can be shown that Theorem
\ref{th:gpc_result} remains valid. The only difference is that the
corresponding DE equations now depend on the schedule through
\begin{align}
	\label{eq:de_mod_z}
	z_i^{(\ell)} = 
	\begin{cases}
		\text{RHS of }\eqref{eq:z_ell_definition} &\quad \text{ if $i \in \mathcal{A}^{(\ell)}$} \\
		z_i^{(\ell-1)} &\quad \text{ otherwise}
	\end{cases},
\end{align}
and
\begin{align}
	\label{eq:de_mod_x}
	x_i^{(\ell)} = 
	\begin{cases}
		\text{RHS of }\eqref{eq:de_gpc} &\quad \text{ if $i \in \mathcal{A}^{(\ell)}$} \\
		x_i^{(\ell-1)} &\quad \text{ otherwise}
	\end{cases}.
\end{align}
To see this first observe that in the proof of
Theorem~\ref{th:gpc_result}, the decoding schedule can be handled by
simply assuming an appropriately modified neighborhood function
$\tilde{\mathcal{D}}_\ell$. In particular, one may think about
embedding the decoding schedule $(\mathcal{A}^{(l)})_{l \in [\ell]}$
into the function $\tilde{\mathcal{D}}_\ell$. Observe that the
scheduling does not change the fact that the decoding outcome is
isomorphism invariant, as long as the type of all vertices is
preserved. Thus, it remains to show that applying the modified
decoding function $\tilde{\mathcal{D}}_\ell$ to the branching process
$\mathfrak{X}$ results in \eqref{eq:de_mod_z} and \eqref{eq:de_mod_x}.
Assuming that the root vertex is active in the final iteration $\ell$,
we can proceed as before. If, on the other hand, the root vertex is
not active in the final iteration $\ell$, we know that the survival
probability is the same as it was in the previous iteration. This
gives \eqref{eq:de_mod_z} and applying the same reasoning for
offspring vertices gives \eqref{eq:de_mod_x}. 

\subsection{Performance on the Binary Symmetric Channel}

When assuming transmission over the \gls{bsc} as opposed to the
\gls{bec}, the crucial difference is that there is a possibility that
the component decoders may miscorrect, in the sense that they
introduce additional errors into the iterative decoding process. This
makes a rigorous analysis challenging. 

One possible approach is to change the iterative decoder. In
particular, consider again the message-passing interpretation of the
iterative decoding in Remark \ref{rmk:message_passing}. In
\cite{Jian2012}, the authors propose to modify the decoder in order to
make the corresponding message-passing update rules extrinsic. In this
case, miscorrections can be rigorously incorporated into the
asymptotic decoding analysis for \gls{gpc} ensembles. The reason why
this approach works from a \gls{de} perspective is that for code
ensembles, the entire computation graph (for a fixed depth) of a CN in
the Tanner graph becomes tree-like. In fact, this makes it possible to
analyze a variety of extrinsic message-passing decoders for a variety
of different channels \cite{Richardson2001}, including the above
modified iterative decoder for the \gls{bsc}. 

Unfortunately, this approach appears to be limited to code ensembles.
Recall that for the deterministic \gls{gpc} construction $\Gpc$, it is
only the neighborhood in the residual graph that becomes tree-like
(not the entire computation graph). Therefore, the independence
assumption between messages is not necessarily satisfied, neither for
intrinsic nor extrinsic message-passing algorithms. In general, it is
not obvious how to rigorously incorporate miscorrections into an
asymptotic analysis for a deterministic \gls{gpc} construction.
Applying our results to the \gls{bsc} thus requires a similar
assumption as in \cite{Schwartz2005,Justesen2007,Justesen2011}, i.e.,
either one assumes that miscorrections are negligible or that a genie
prevents them. 

\RevA{%
\subsection{Spatially-Coupled Product Codes}%
\label{sec:spatially_coupled}%

Of particular interest are GPCs where the matrix $\etab$ has a
band-diagonal ``convolutional-like'' structure. The associated GPC can
then be classified as a spatially-coupled PC. For example, the GPCs
discussed in Examples \ref{ex:staircase_codes} and
\ref{ex:braided_codes}, i.e., staircase and braided codes, are
particular instances of spatially-coupled PCs.  Spatially-coupled
codes have attracted a lot of attention in the literature due to their
outstanding performance under iterative decoding \cite{Kudekar2011,
Yedla2014}.  

In \cite{Haeger2016isit}, we study the asymptotic performance of
deterministic spatially-coupled PCs based on the theory derived in
this paper. In particular, we provide a detailed comparison to
spatially-coupled PCs that are based on the code ensembles proposed in
\cite{Jian2012, Jian2015, Zhang2015}. One of the main outcomes of this
work is that there exists a family of deterministic spatially-coupled
PCs (see \cite[Def.~2]{Haeger2016isit}) that asymptotically follows
the same DE recursion as the ensembles defined in \cite{Jian2012,
Jian2015, Zhang2015}. This implies that certain ensemble properties
proved in these papers (in particular lower bounds on the decoding
thresholds via potential function methods) also apply to the
deterministic code family as well. An important step to show this
result is to transform the ensemble DE recursions obtained in
\cite{Jian2012, Jian2015, Zhang2015} into a form that makes them
comparable to the DE recursion for deterministic GPCs obtained in
this paper.
We also show that there exists a
related, but structurally simpler, deterministic code family (see
\cite[Def.~5]{Haeger2016isit}) that attains essentially the same
asymptotic performance. The simpler code family follows a slightly
different DE recursion, which can also be analyzed using potential
function methods (see the appendix of \cite{Haeger2016isit} for
details).

It is also interesting to compare the asymptotic DE predictions to
performance of practical finite-length spatially-coupled PCs. This was
done in \cite{Haeger2016ofc} for staircase, braided, and so-called
half-braided codes assuming a window decoder with realistic
parameters. }%

\NoRev{%
\subsection{Parallel Binary Erasure Channels}%
\label{sec:parallel_becs}

It is possible to generalize the analysis presented in this paper to
the case where transmission takes place over $M$ parallel BECs with
different erasure-correcting capabilities $p_1, \dots, p_M$
\cite{Haeger2016istc}. In that scenario, one considers the scaling
$p_1 = c_1/n, \dots, p_M = c_M/n$, where the positive constants $c_1,
\dots, c_M$ act as the effective channel qualities for the parallel
BECs. As an application, the authors showed in \cite{Haeger2016istc}
that the resulting asymptotic performance analysis can be used to
predict and optimize the performance of $\Gpc$ when combined with a
higher-order signal constellation in a coded modulation setup assuming
a hard-decision symbol detector.}%

\section{Irregular Half-Product Codes}
\label{sec:irregular_hpc}

In this section, we consider an application of the derived \gls{de}
equations for deterministic \glspl{gpc}. In particular, we discuss the
optimization of component code mixtures for \glspl{hpc}. Recall that
for (regular) \glspl{hpc}, we have $\etab = 1$, $\vect{\gamma} = 1$,
and all component codes associated with the CNs have the same
erasure-correcting capability $\tee$. Similar to irregular \glspl{pc}
\cite{Hirasawa1984, Alipour2012}, an \emph{irregular} \gls{hpc} is
obtained by assigning component codes with different
erasure-correcting capabilities to the CNs. \RevA{The primary goal of
this section is to show how the asymptotic DE analysis can be used in
practice to achieve performance improvements for finite-length codes.
The general approach is to use decoding thresholds as an optimization
criterion. This is, of course, completely analogous to optimizing
degree distributions of irregular LDPC codes based on DE. Thus, it comes
with similar caveats (e.g., no optimality guarantees for finite code
lengths), but also with similar strengths (e.g., an efficient and fast
optimization procedure).}


\subsection{Preliminaries}

The assignment of erasure-correcting capabilities to the \CNs is done
according to the distribution $\vect{\tau} = (\tau_1, \dots,
\tau_{\tmax})^\transpose$.
(For notational convenience, we suppress the dependence of the
distribution and other quantities on the position index in the Tanner
graph.) The mean erasure-correcting capability \eqref{eq:average_tee}
in this case is given by
\begin{align}
	\bar{\tee} = \sum_{\tee = 1}^{\tmax} \tau_\tee \tee. 
\end{align}
The \gls{de} equation \eqref{eq:de_gpc} simplifies to
\begin{align}
	\label{eq:irr_hpc_de}
	x^{(\ell)} = \sum_{\tee=1}^{\tmax} \tau_{\tee} \Psi_{\geq \tee} (c
	x^{(\ell-1)}), 
\end{align}
with $x^{(0)} = 1$. The decoding threshold \eqref{eq:gpc_threshold}
can alternatively be written as
\begin{align}
	\label{eq:irregular_hpc_threshold}
	\cthr = \sup \{ c > 0 \,|\, \lim_{\ell \to \infty}
	x^{(\ell)} = 0
	\}, 
\end{align}
since $z^{(\ell)} \to 0$ if and only if $x^{(\ell)} \to 0$ as $\ell
\to \infty$. From \eqref{eq:irr_hpc_de} and the fact that $\Psi_{\geq
\tee} (x)$ for any $\tee \in \mathbb{N}$ and $x \geq 0$ is strictly
increasing, we have that the condition
\begin{align}
	\label{eq:irr_hpc_condition}
	\sum_{\tee=1}^{\tmax} \tau_{\tee} \Psi_{\geq \tee} (c
	x) < x, \qquad \text{for $x \in (0, 1]$}, 
\end{align}
implies successful decoding after a sufficiently large number of
iterations, i.e., we have that $\cthr \geq c$. 

We wish to design $\vect{\tau}$ such that $\cthr$ is as large as
possible.  Obviously, choosing component codes with larger
erasure-correcting capability gives better performance, i.e., larger
thresholds. Thus, the design is done under the constraint that the
mean erasure-correcting capability $\bar{\tee}$ remains fixed. This is
the natural analogue to the rate-constraint when designing degree
distributions for irregular \gls{ldpc} codes. 

\subsection{Lower Bounds on the Threshold}

Before discussing the practical optimization of the distribution
$\vect{\tau}$ based on a linear program in the next subsection, we
show that one can construct irregular \glspl{hpc} that have thresholds 
\begin{align}
	\label{eq:threshold_bounds}
	2 \bar{\tee} - 1 \leq \cthr \leq 2 \bar{\tee}, 
\end{align}
where we recall that the upper bound in \eqref{eq:threshold_bounds}
holds for any GPC according to the discussion in Section
\ref{sec:upper_bound}.  The lower bound in \eqref{eq:threshold_bounds}
is achieved by a uniform distribution. In particular, from
$\sum_{i=1}^{\infty} \Pr{X \geq i } = \E{X}$, we have 
\begin{align}
	\sum_{\tee = 1}^{\infty} \Psi_{\geq \tee}(c x) = c x,
\end{align}
where we recall that $\Psi_{\geq \tee}(c x) = \Pr{\Pois{c x} \geq
\tee}$. If we then choose a uniform distribution according to
$\tau_{\tee}
= 1/N$ for $\tee \in [N]$ (i.e., $\tmax = N$), we have
\begin{align}
	\label{eq:unif_conv_cond}
	\sum_{\tee = 1}^{N} \tau_\tee \Psi_{\geq \tee}(N x) 
	< 
	\sum_{\tee = 1}^{\infty} \frac{1}{N} \Psi_{\geq \tee}(N x) 
	= x \quad \text{for $x > 0$}, 
\end{align}
where the (strict) inequality follows from the fact that $\Psi_{\geq
\tee}(x) > 0$ for any $\tee \in \mathbb{N}$ and $x > 0$.  We see from
\eqref{eq:unif_conv_cond} that the threshold for the uniform
distribution satisfies $\cthr \geq N$
(cf.~\eqref{eq:irr_hpc_condition}).  Moreover, the average
erasure-correcting capability is given by
\begin{align}
	\bar{\tee} = \sum_{\tee = 1}^{N} \tau_\tee \tee
	= \frac{1}{N} \frac{N (N+1)}{2}
	= \frac{N+1}{2}.
\end{align}
Therefore, we have
\begin{align}
	2 \bar{\tee} - \cthr \leq 2 \frac{N + 1}{2} - N = 1, 
\end{align}
or $\cthr \geq 2 \bar{\tee} - 1$. This simple lower bound shows that one
can design irregular \glspl{hpc} that are within a constant gap of the
upper $2\bar{\tee}$-bound. This is in contrast to regular \glspl{hpc}
where $\bar{\tee} = \tee$. In this case, the difference between the
threshold $\cthr$ and $2 \tee$ becomes unbounded for large $\tee$, since
$\cthr = \tee + \sqrt{\tee \log \tee} + \mathcal{O}(\log(\tee))$
\cite{Pittel1996}. 

\begin{remark}
Essentially the same argument also allows us to give a lower bound on
the threshold for irregular \glspl{hpc} when the minimum
erasure-correcting capability is constrained to some value $\tmin >
1$. In that case, a uniform distribution over $\{\tmin, \tmin + 1,
\dots, \tmin + N - 1\}$ still gives a threshold that satisfies $\cthr
\geq N$. However, we have $\bar{\tee} = (N + 2 \tmin - 1)/2$. 
Hence, one obtains the lower bound $\cthr \geq 2 \bar{\tee} - 2 \tmin +
1$. 
\end{remark}

\subsection{Optimization via Linear Programming}

The optimal distribution maximizes the threshold $\cthr$ subject to a
fixed mean erasure-correcting capability $\bar{\tee}$. Alternatively,
one may fix a certain channel quality parameter $c$ and minimize
$\bar{\tee}$ as follows. 
\begin{align}
		\underset{\tau_1, \dots, \tau_{\tmax}}{\text{minimize}} \quad
		& \bar{\tee} = \sum_{\tee = 1}^{\tmax} \tau_\tee \tee  \label{eq:objective}\\ 
		\text{subject to} \quad & \sum_{\tee = 1}^{\tmax} \tau_\tee = 1,
		\quad \tau_1, \dots, \tau_{\tmax} \geq 0  \label{eq:constraint1} \\
		& \sum_{\tee = 1}^{\tmax} \tau_\tee \Psi_{\geq \tee} (c x) < x, \label{eq:constraint2}
		\quad x
		\in (0, 1].
\end{align}
The objective function and all constraints in
\eqref{eq:objective}--\eqref{eq:constraint2} are linear in $\tau_1,
\dots, \tau_{\tmax}$. Thus, after discretizing the constraint
\eqref{eq:constraint2} according to $x = i \Delta$ for $i \in [M]$ and
$\Delta = 1/M$, one obtains a linear program, which can be efficiently
solved by standard numerical optimization solvers. In
Fig.~\ref{fig:hpc_threshold_comparison}, we show the thresholds of the
optimized irregular HPCs by the red line, where we used $M = 1000$ and
$\tmax = 50$, as a function of $\bar{\tee}$. We also show the
thresholds for regular \glspl{hpc} (where $\bar{\tee} = \tee = 2, 3,
\dots$) and the $2\bar{\tee}$-bound by the blue and black lines,
respectively. It can be seen that the thresholds for regular
\glspl{hpc} diverge from the bound for large $\bar{\tee}$, as
expected.  Using irregular \glspl{hpc}, the thresholds can be
significantly improved for large $\bar{\tee}$.  However, there appears
to be an almost constant gap between the upper bound and the threshold
curve. This gap is investigated in more detail in the next subsection. 

For practical applications, it is often desirable to limit the
fraction of component codes with ``small'' erasure-correcting
capabilities in order to avoid harmful error floors
\cite{Justesen2011}. It is straightforward to incorporate a minimum
erasure-correcting capability $\tmin$ into the above linear program.
For example, the green line in Fig.~\ref{fig:hpc_threshold_comparison}
shows the thresholds of the optimized irregular HPCs when the minimum
erasure-correcting capability is constrained to $\tmin = 4$. This
additional constraint entails a threshold penalty which, however,
decreases for larger values of $\bar{\tee}$. 

\begin{figure}[t]
	\begin{center}
		\includegraphics{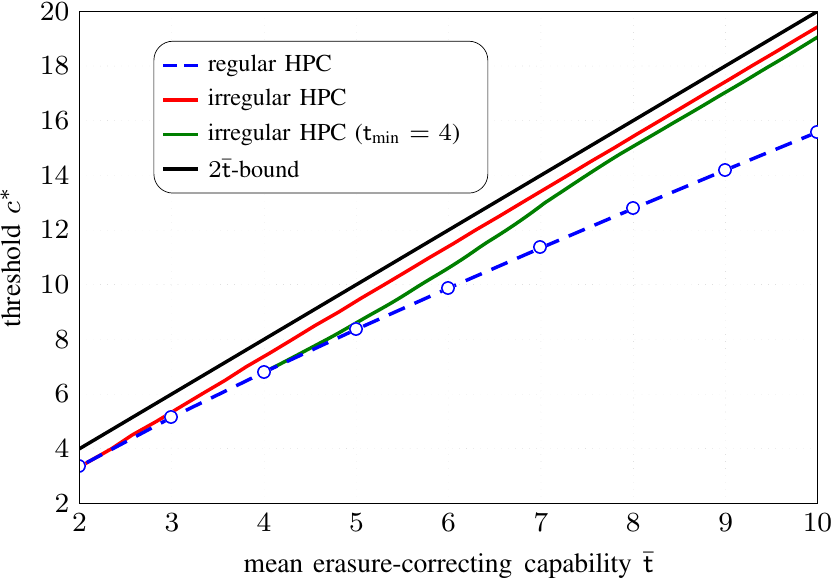}
	\end{center}
	\caption{Decoding thresholds for optimized irregular and regular
	\glspl{hpc}. Thresholds for irregular HPCs are obtained via a
	discretized linear program with $M = 1000$ and $\tmax = 50$.}
	\label{fig:hpc_threshold_comparison}
\end{figure}

\subsection{Initial Component Code Loss}

We now focus in more detail on the upper $2\bar{\tee}$-bound for the
thresholds of irregular \glspl{hpc}. In particular, we show that it is
possible to give a slightly improved upper bound based on the notion
of an \emph{initial component code loss}. Based on this, it can be
shown that the upper bound in \eqref{eq:threshold_bounds} is in fact
strict, i.e., for any distribution $\vect{\tau}$ with mean
erasure-correcting capability $\bar{\tee}$ and threshold $\cthr$, the
gap $2\bar{\tee} - \cthr$ is always bounded away from zero. This gives
an intuitive explanation for the gap between the threshold curve and
the $2\bar{\tee}$-bound observed in
Fig.~\ref{fig:hpc_threshold_comparison}. A similar bound for irregular
\gls{ldpc} code ensembles over the \gls{bec} is given in
\cite{Shokrollahi1999}.

Recall that the upper bound has been derived in
Section~\ref{sec:upper_bound} under the (somewhat optimistic) assumption
that each $\tee$-erasure correcting component code corrects exactly
$\tee$ erasures. In other words, each component code is assumed to
contribute its maximum erasure-correcting potential to the overall
decoding. A refined version of this argument takes into account the
fact that a certain amount of erasure-correcting potential is lost
almost surely before the iterative decoding process even begins. In
particular, let the RVs $N_{i, \tee}$, for $i = 0, 1 \dots, \tee - 1$,
be the number of CNs corresponding to $\tee$-erasure-correcting
component codes that are initially connected to $i < \tee$ erased VNs.
In the first decoding iteration, each of these CNs corrects only $i$
erasures instead of $\tee$, i.e., the maximum number of erasures $E$
that we can hope to correct is upper bounded by
\begin{align}
	E \leq \bar{\tee} n - \sum_{\tee = 1}^{\tmax} \sum_{i=0}^{\tee-1}
	N_{i, \tee} (\tee - i).
\end{align}
Since $E/n$ and $N_{i, \tee}/n$ converge almost surely to the
deterministic values $c/2$ and $\tau_\tee \Psi_{=i}(c)$, respectively,
we obtain 
\begin{align}
	\label{eq:loss_condition}
	c \leq 2 \bar{\tee} - 2 \Loss_{\vect{\tau}}(c)
\end{align}
as a necessary condition for successful decoding, where we implicitly
defined the initial component code loss for the distribution
$\vect{\tau}$ as
\begin{align}
	\Loss_{\vect{\tau}}(c) \define \sum_{\tee = 1}^{\tmax} \tau_\tee \Loss(\tee, c) 
\end{align}
with
\begin{align}
	\Loss(\tee, c) \define \sum_{i=0}^{\tee-1} \Psi_{=i}(c) (\tee -
	i)
\end{align}
for $c > 0$ and $\tee \in \mathbb{N}$. 

\begin{remark}
The affine extension of $\Loss(\tee, c)$ for a fixed
$c \geq 0$ is convex in $\tee \in [1; \infty)$ in the sense
that for any $c \geq 0$ and $\tee = 2, 3, \dots$, we have
\begin{align}
	\Loss(\tee-1, c) + 
	\Loss(\tee+1, c) &= 
	2 \Loss(\tee, c) + \Psi_{=t}(c) \\
	&\geq 2 \Loss(\tee, c).
\end{align}
This implies that for any distribution $\vect{\tau}$ with average
erasure-correcting capability $\bar{\tee}$, the associated initial
component code loss satisfies
\begin{align}
	\Loss_{\vect{\tau}} (c) \geq \Loss(\lfloor \bar{\tee} \rfloor, c),
\end{align}
i.e., the initial loss is minimized for regular \glspl{hpc}.
\end{remark}

The bound \eqref{eq:loss_condition} has a natural interpretation in
terms of areas related to the curves involved in the condition
\eqref{eq:irr_hpc_condition}, similar to the area theorem for
irregular \gls{ldpc} code ensembles. Indeed, an alternative way to
show that successful decoding implies \eqref{eq:loss_condition} is by
integrating the condition \eqref{eq:irr_hpc_condition}. Using integration
by parts, one obtains the indefinite integral \cite{Gradshteyn2007}
\begin{align}
	\label{eq:indefinite_integral_psi}
	\int \Psi_{\geq \tee}(x) \, \D x = x \Psi_{\geq \tee} (x) + \tee
	\Psi_{\leq \tee} (x).
\end{align}
Thus, we have
\begin{align}
	c \int_0^1 \Psi_{\geq \tee} (c x) \, \D x 
	&= c \Psi_{\geq \tee}(c) + \tee \Psi_{\leq \tee}(c) - \tee \\
	&= c (1 - \Psi_{< \tee}(c)) + \tee \Psi_{\leq \tee}(c) - \tee \\
	&= c - \tee + \Loss(\tee, c),
	\label{eq:integral_psi}
\end{align}
where the last equality follows from
\begin{align}
	\tee \Psi_{\leq \tee}(c) - c \Psi_{< \tee} (c)
	&=\tee \Psi_{\leq \tee}(c) - c\sum_{k=0}^{\tee-1} \frac{c^k}{k!} e^{-c}\\
	&=\tee \Psi_{\leq \tee}(c) - c \sum_{k=1}^{\tee}
	\frac{c^{k-1}}{(k-1)!}e^{-c} \\
	&=\tee \Psi_{\leq \tee}(c) - \sum_{k=0}^{\tee}
	\frac{c^{k}}{k!}e^{-c} k \\
	&=\sum_{k=0}^{\tee} \Psi_{= k}(c) (\tee - k) \\
	&= \Loss(\tee, c).
\end{align}
Hence, integrating both sides of \eqref{eq:irr_hpc_condition} from
zero to one and using \eqref{eq:integral_psi}, one obtains
\begin{align}
	\label{eq:condition_inter}
	\frac{1}{c} \sum_{\tee=1}^{\tmax} \tau_\tee \left(c - \tee +
	\Loss(\tee, c)  \right) < \frac{1}{2}, 
\end{align}
or, equivalently, \eqref{eq:loss_condition}.

\begin{figure}[t]
	\begin{center}
		\includegraphics{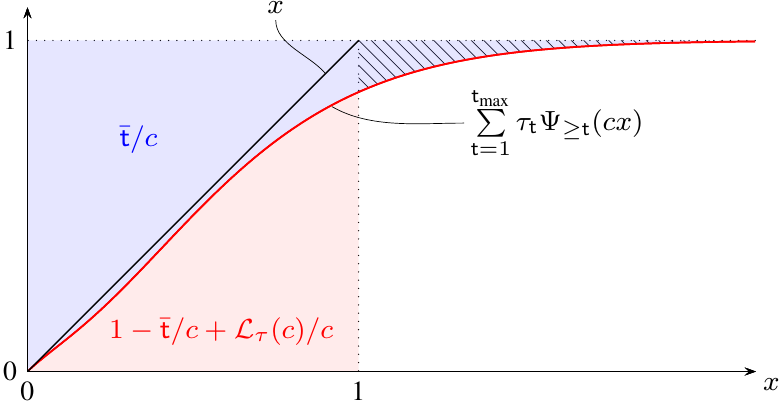}
	\end{center}
	\caption{Graphical interpretation of the upper threshold bounds. }
	\label{fig:area_theorem}
\end{figure}

A visualization is shown in Fig.~\ref{fig:area_theorem}, where the
red and black lines correspond to the \gls{lhs} and \gls{rhs}
of \eqref{eq:irr_hpc_condition}, respectively. The area below the red
curve up to $x=1$ (shown in red) corresponds to the \gls{lhs} of
\eqref{eq:condition_inter}. Similarly, it can be shown using
\eqref{eq:indefinite_integral_psi} that the area between the red line
and $x=1$ (shown in blue) corresponds to the scaled erasure-correcting
capability $\bar{\tee}/c$. Note that the $2\bar{\tee}$-bound on the
threshold simply corresponds to the fact that the blue area cannot be
smaller than $1/2$, since otherwise the red and black lines would
have to cross. From the previous discussion, we have seen that the gap
to the upper $2\bar{\tee}$-bound is partially due to the initial
component code loss. In particular, by combining the blue and red
areas, it can be seen that the hatched area in
Fig.~\ref{fig:area_theorem} corresponds precisely to the (scaled) loss
$\Loss_{\vect{\tau}}(c)/c$. 

Consider now again the outcome of the linear program for the optimized
irregular \glspl{hpc} in Fig.~\ref{fig:hpc_threshold_comparison}. In
Fig.~\ref{fig:irregular_gap_to_cap}, the (vertical) gap $2 \bar{\tee}
- \cthr$ between the black and red lines in
Fig.~\ref{fig:hpc_threshold_comparison} is shown for a larger range of
$\bar{\tee}$. It can be seen that the gap is decreasing with
$\bar{\tee}$, albeit rather slowly. We also plot the initial component
code loss for the optimized distributions at the threshold value by
the blue line. From this, we see that the initial component code loss
accounts for approximately half of the threshold gap for the optimized
irregular distributions. 

\begin{figure}[t]
	\begin{center}
		\includegraphics{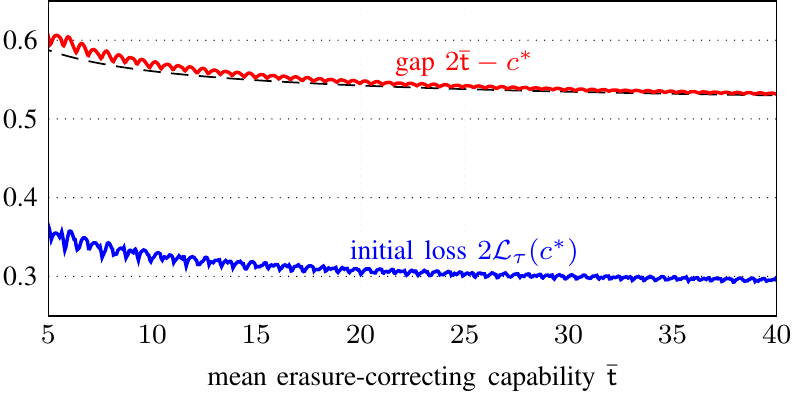}
	\end{center}
	\caption{Gap between the threshold $\cthr$ and the
	$2\bar{\tee}$-bound for the optimized irregular \glspl{hpc}. }
	\label{fig:irregular_gap_to_cap}
\end{figure}

\begin{remark}
	In fact, we conjecture that the following is true. Assume $c \in
	\mathbb{N}$. Then, for any distribution $\vect{\tau}$ with threshold $\cthr
	\geq c$ and mean erasure-correcting capability $\bar{\tee}$, we have
	\begin{align}
		\label{eq:bound_conjecture}
		\bar{\tee} \geq \frac{c}{2} + \frac{1}{c} \sum_{\tee = 1}^c
		\Loss(\tee, c). 
	\end{align}
	This bound is shown in Fig.~\ref{fig:irregular_gap_to_cap} by the
	dashed line, although we failed to prove it. Proving
	\eqref{eq:bound_conjecture} would be interesting, since one can show
	that $\lim_{c \to \infty} \frac{1}{c} \sum_{\tee =1}^{c} \Loss(\tee,
	c) = 1/4$ and hence $2 \bar{\tee} - \cthr \geq 1/2$, which seems to be
	the constant to which the optimization outcome is converging for
	$\bar{\tee} \to \infty$. 
\end{remark}

\subsection{Simulation Results}

In order to illustrate how the thresholds can be used to design
practical irregular \glspl{hpc}, we consider (shortened) binary
\gls{bch} codes as component codes. Given the Galois-field extension
degree $\nu$, a shortening parameter $s$, and the erasure-correcting
capability $\tee$, we let the component code be an $(\nbch, \kbch,
\dmin)$ \gls{bch} code, where $\nbch = 2^{\nu} - 1 - s$, $\dmin = \tee
+ 1$, and
\begin{align}
	\kbch = \begin{cases}
		\nbch - \nu \tee/2, & \tee \text{ even} \\
		\nbch - \nu (\tee-1)/2 - 1, & \tee \text{ odd} \\
	\end{cases}.
\end{align}
In the following, we consider two irregular \glspl{hpc}, where
$\bar{\tee} \approx 7$. As a comparison, we use a regular \gls{hpc}
with $\tau_7 = 1$ for which $\cthr \approx 11.34$. The optimal
distribution (rounded to three decimal places) according to the linear
program \eqref{eq:objective}--\eqref{eq:constraint2} is given by
\begin{equation}
	\label{eq:distribution1}
	\begin{aligned}
	\tau_{1} &= 0.070, \quad
	\tau_{2} = 0.103, \quad
	\tau_{4} = 0.115, \\
	\tau_{5} &= 0.179, \quad
	\tau_{10} = 0.496, \quad
	\tau_{11} = 0.037,
	\end{aligned}
\end{equation}
which yields $\cthr \approx 13.42$. We also consider the case where
the minimum erasure-correcting capability is constrained to be $\tmin
= 4$. For this case, one obtains
\begin{align}
	\label{eq:distribution2}
	\tau_{4} = 0.495, \quad
	\tau_{9} = 0.029, \quad
	\tau_{10} = 0.476,
\end{align}
and the threshold is reduced to $\cthr \approx 12.88$.

For the simulations, we consider two different component code lengths,
$\nbch = 1000$ (i.e., $\nu = 10$ and $s = 23$) and $\nbch = 3000$
(i.e., $\nu = 12$ and $s = 1095$), leading to an overall length of the
\glspl{hpc} of $m \approx 500,000$ and $m \approx 4,500,000$,
respectively. If we denote the dimension of the $k$-th component code
by $k_{\mathcal{C}_k}$, the code rate is lower bounded by
\cite[Sec.~5.2.1]{Ryan2009}
\begin{align}
	R \geq 
	1 - \frac{\sum_{k = 1}^{n} (n - k_{\mathcal{C}_k})}{m}.
\end{align}
For the regular case and the distributions \eqref{eq:distribution1}
and \eqref{eq:distribution2}, the lower bound evaluates to
approximately $0.93$ and $0.97$ for $\nbch = 1000$ and $\nbch = 3000$,
respectively.  (In order to obtain shorter (longer) codes for the same
rate, one needs to reduce (increase) $\bar{\tee}$.) Although the
chosen values for $n$ are merely for illustration purposes, we remark
that the delay caused by such seemingly long block-lengths is
typically not a problem for high-speed applications. For example, the
delay for the \glspl{gpc} designed for fiber-optical communication
systems in \cite{Smith2012a, Jian2014} is in the order of $2,000,000$
bits. 

\begin{figure}[t]
	\begin{center}
		\includegraphics{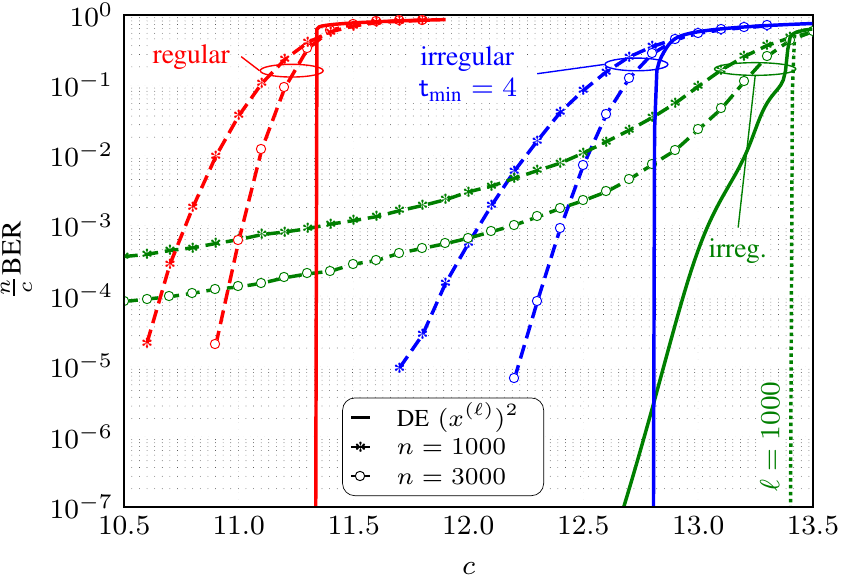}
	\end{center}
	\caption{Simulation results (dashed lines) for regular and optimized
	irregular \glspl{hpc} for two different values of $n$ and $\ell =
	100$. DE results (solid lines) are shown for $\ell = 100$.}
	\label{fig:hpc_simulation}
\end{figure}

Simulation results are shown in Fig.~\ref{fig:hpc_simulation} by the
dashed lines. In all cases, the maximum number of decoding iterations
is restricted to $\ell = 100$. Results for regular \glspl{hpc} are
shown in red, while results for the irregular \glspl{hpc} defined by
the optimized distributions \eqref{eq:distribution1} and
\eqref{eq:distribution2} are shown in green and blue, respectively.
For lower error rates, the irregular HPCs defined by
\eqref{eq:distribution1} are clearly outperformed by regular HPCs and
HPCs defined by the distribution \eqref{eq:distribution2}. This is due
to the relatively large fraction of component codes that only correct
1 and 2 erasures, which leads to a large error floor. 

It is interesting to inspect the DE predictions for $\ell =100$, which
are shown by the solid lines in Fig.~\ref{fig:hpc_simulation}. The
predicted performance for the regular and irregular distribution
\eqref{eq:distribution2} drops sharply, while the predicted
performance for the distribution \eqref{eq:distribution1} shows a
markedly different behavior due to the finite iteration number. It is
therefore important to stress that an optimization via the condition
\eqref{eq:irr_hpc_condition} implicitly assumes an unrestricted number
of decoding iterations. (As a reference, the DE prediction for the
distribution \eqref{eq:distribution1} with $\ell = 1000$ is shown by
the green dotted line.) Thus, if we had done an optimization based on
DE assuming $\ell = 100$ and targeting an error rate of around
$10^{-7}$ in Fig.~\ref{fig:hpc_simulation}, we would have rejected the
distribution \eqref{eq:distribution1} in favor of the distribution
\eqref{eq:distribution2} right away. \NoRev{A method to incoorporate the
number of decoding iterations into the threshold optimization of
irregular LDPC code ensembles is discussed for example in
\cite{Smith2010a}.}


Lastly, the HPCs defined by \eqref{eq:distribution2} have a comparable
finite-length scaling behavior below the threshold and no noticeable
error-floor for the simulated error rates. As a consequence, the
performance gains for this distribution over the regular HPCs
predicted by DE are well preserved also for finite lengths. 

\section{Conclusions and Future Work}
\label{sec:conclusion}

In this paper, we studied the performance of deterministically
constructed \glspl{gpc} under iterative decoding. Using the framework
of sparse inhomogeneous random graphs, we showed how to derive the
\gls{de} equations that govern the asymptotic behavior. In principle,
DE can be used for a variety of different applications, e.g., parameter
tuning, optimization of decoding schedules, or the design of new GPCs.
Here, we used the derived DE equations to optimize irregular HPCs that
employ a mixture of component codes with different erasure-correcting
capabilities. Using an approach based on linear programming, we
obtained irregular HPCs that outperform regular HPCs. 

For future work, it would be interesting to analyze deterministic code
constructions that incorporate VNs with larger degrees. Larger VN
degrees are easily incorporated into an ensemble approach, see, e.g.,
\cite{Zhang2015}. An example of a corresponding deterministic code
construction is the case where code arrays are generalized from two to
three (or higher) dimensional objects, e.g., a cube-shaped code array.
In that case, the residual graph could be modeled as a random
hypergraph. Cores in random hypergraphs have for example been studied
in \cite{Molloy2005}. 

\RevA{Another interesting topic for future work is the investigation
of the finite-length scaling behavior of deterministic \glspl{gpc},
similar to the work for LDPC codes in \cite{Amraoui2009}. In fact, the
scaling behavior of the $k$-core in random graphs is characterized in
\cite{Janson2008} and it would be interesting to translate these
results into the coding setting. To the best of our knowledge,
finite-length scaling for the inhomogeneous random graph model has not
yet been considered. However, such an analysis would be of great
practical value since it may allow for accurate performance
estimations of GPCs under iterative BDD for finite values of $n$.}

\appendices

\section{Proof of Lemma \ref{lem:maximum_vertex_degree}}
\label{app:maximum_vertex_degree}

First, we upper-bound the probability that the degree $D_k$ of the
$k$-th vertex exceeds $d_n$ using the Chernoff bound. Let $\bRV \sim
\Bern{c/n}$, then, for any $\lambda > 0$, we have
\begin{align}
	\Pr{D_k \geq d_n} &= \Pr{e^{\lambda D_k} \geq e^{\lambda 
	d_n }} \\
	&\leq e^{-\lambda d_n} \E{e^{\lambda D_k}}\label{eq:maximum_vertex_degree_markov} \\
	&=e^{-\lambda d_n}  \E{e^{\lambda (\bRV_{k,1} + \dots + \bRV_{k,n})}}  \\
	&=e^{-\lambda d_n}  \left( \E{e^{\lambda \bRV}}
	\right)^{n-1} \label{eq:maximum_vertex_degree_independence} \\
	&=e^{-\lambda d_n}   (1 - p + p e^{\lambda})^{n-1} \\
	&\leq e^{-\lambda d_n}  \left(1 +
	\frac{c}{n}\left(e^{\lambda} - 1\right) \right)^{n} \\
	&\leq e^{-\lambda
	d_n} e^{c (e^\lambda - 1)} \label{eq:max_degree_tmp} \\ 
	&\leq e^{-c - d_n \ln \frac{d_n}{ce}} \label{eq:max_degree_chernoff}
\end{align}
where \eqref{eq:maximum_vertex_degree_markov} follows from applying
Markov's inequality, \eqref{eq:maximum_vertex_degree_independence}
holds because all $\bRV_{k,j} \sim \bRV$ are independent except
$\bRV_{k,k} = 0$, \eqref{eq:max_degree_tmp} stems from $(1 + x/n)^n
\leq e^x$ for $x\geq 0$, and \eqref{eq:max_degree_chernoff} follows
from minimizing over $\lambda$. Thus, for $d_n = \Omega(\log(n))$ and
any $\beta >0$, there is an $n_0$ such that $\Pr{D_k \geq d_n} \leq
e^{-\beta d_n}$. Hence, if one chooses $\beta$ large enough, then the
union bound implies 
\begin{align}
	\Pr{D_\text{max} \geq d_n} 
	&\leq n  \Pr{D_k \geq d_n} \\
	&\leq n e^{-\beta d_n} \\ 
	&= e^{\log(n) - \beta d_n} \\
	&= e^{-\beta (d_n - \log(n) / \beta)} \\
	&\leq e^{-\beta d_n/2}
\end{align}
for all $n\geq n_0$. 

\section{Bound on the Second Moment of $T_\ell$}
\label{app:moments_stopping_times}

\newcommand{\mean}{\mu_{\bar{\xi}}}
\newcommand{\varxi}{\sigma_{\bar{\xi}}}

\RevA{To obtain a bound on $\mathbb{E}[T_\ell^2]$, we first show how
to compute the corresponding quantity $\mathbb{E}[\GwT_\ell^2]$ for
the branching process. This quantity depends only on the mean $\mean$
and variance $\varxi^2$ of the offspring distribution $\GwXi$.
Then, we apply the result that the exploration process is
stochastically dominated by a Poisson branching process with $\GwXi =
\Pois{c}$; in particular, the random variable $T_\ell^2$ is
stochastically dominated by the random variable $\GwT_\ell^2$. (Recall
that if $Y$ stochastically dominates $X$, we have $\mathbb{E}[X] \leq
\mathbb{E}[Y]$, see, e.g., \cite[Sec.~2.3]{Hofstad2014}.)}

First, from $\GwT_\ell = \GwZ_0 + \GwZ_1 + \dots \GwZ_\ell$, we obtain
\begin{align}
	\label{eq:t2_def}
	\mathbb{E}[\GwT_\ell^2] = \sum_{i = 0}^{\ell} \mathbb{E}[\GwZ_i^2] + 2 \sum_{i=1}^{\ell}
	\sum_{j=0}^{i-1} \mathbb{E}[\GwZ_i \GwZ_j]. 
\end{align}
Using the definition of $\GwZ_i$ and the law of total expectation, it
can be shown that for $i > j$, we have
\begin{align}
	\label{eq:t2_tmp1}
	\mathbb{E}[\GwZ_i \GwZ_j] = \mean^{i-j} \mathbb{E}[\GwZ_j^2]. 
\end{align}
Inserting \eqref{eq:t2_tmp1} into \eqref{eq:t2_def} leads to 
\begin{align}
	\label{eq:t2_tmp2}
	\mathbb{E}[\GwT_\ell^2] = \sum_{i = 0}^{\ell} \mathbb{E}[\GwZ_i^2] + 2 \sum_{i=1}^{\ell}
	\sum_{j=0}^{i-1} \mean^{i-j}\mathbb{E}[\GwZ_j^2].
\end{align}
Next, we can use the well-known expressions for the mean and variance
of $\GwZ_i$ (see, e.g., \cite[p.~396]{Karlin1975}) to obtain
\begin{align}
	\label{eq:second_moment_Z_ell}
	\mathbb{E}[\GwZ_\ell^2] =  
	\begin{cases}
		\varxi^2 \mean^{\ell-1} \frac{\mean^{\ell} - 1}{\mean-1} +
		\mean^{2\ell},
		&\quad \mean \neq 1  \\
		\ell \varxi^2 + 1, &\quad \mean = 1
	\end{cases}.
\end{align}
Inserting \eqref{eq:second_moment_Z_ell} into \eqref{eq:t2_tmp2} leads
to the desired explicit characterization of $\mathbb{E}[\GwT_\ell^2]$.
Of particular interest here is the case where the offspring
distribution is Poisson with mean $c$. In this case, we have $\mean = c$
and $\varxi^2 = c$, which leads to
\begin{align}
	\label{eq:second_moment_T_ell}
	\mathbb{E}[\GwT_\ell^2] =  
	\begin{cases}
		\frac{c^{2\ell+3} - 1 - (2\ell+3) c^{\ell} (c-1) }{(c-1)^3},
		&\quad c \neq 1  \\
		\frac{(\ell+1)(\ell+2)(2\ell+3)}{6}, &\quad c = 1
	\end{cases}.
\end{align}
Finally, using the same steps as in the proof of
\cite[Th.~4.2]{Hofstad2014} and \cite[Th.~3.20]{Hofstad2014} one can
show that \RevA{the random variable $T_\ell^2$ is stochastically
dominated by the random variable $\GwT_\ell^2$. Hence,
\eqref{eq:second_moment_T_ell} is an upper bound on
$\mathbb{E}[T_\ell^2]$, i.e., we have $\mathbb{E}[T_\ell^2] \leq
\mathbb{E}[\GwT_\ell^2]$.}


\end{document}